\documentclass{article} 

\usepackage{arxiv}

\usepackage[utf8]{inputenc} 
\usepackage[T1]{fontenc}    

\usepackage{mathtools}
\usepackage{amsmath}
\usepackage{amsfonts}
\usepackage{amsthm}
\usepackage{color}
\usepackage{url}
\usepackage[hidelinks]{hyperref}
\usepackage{graphicx}
\usepackage{booktabs}
\usepackage{tabularx}
\usepackage{algorithm}
\usepackage{algorithmic}
\usepackage{microtype}
\usepackage{doi}
\usepackage{soul}
\usepackage{xcolor}

\usepackage{subcaption}
\usepackage{latexsym}
\usepackage[square,numbers]{natbib}

\usepackage{framed}
\usepackage[most]{tcolorbox}

\usepackage{lmodern}

\title{Sybil-Proof Diffusion Auction in Social Networks}
\date{November, 2022}

\author{
Hongyin Chen\\
	Center on Frontiers of Computing Studies\\
	Peking University\\
	Beijing, China \\
	\texttt{chenhongyin@pku.edu.cn} \\
	\And
Xiaotie Deng\\
	Center on Frontiers of Computing Studies\\
	Peking University\\
	Beijing, China \\
	\texttt{xiaotie@pku.edu.cn} \\
	\And
Ying Wang\\
	Center on Frontiers of Computing Studies\\
	Peking University\\
	Beijing, China \\
	\texttt{wying2000@pku.edu.cn} \\
	\And
Yue Wu\\
	Center on Frontiers of Computing Studies\\
	Peking University\\
	Beijing, China \\
	\texttt{wu.y@pku.edu.cn} \\
	\And
Dengji Zhao\\
	ShanghaiTech University\\
	Shanghai, China \\
	\texttt{zhaodj@shanghaitech.edu.cn} \\
}
\renewcommand{\shorttitle}{Sybil-Proof Diffusion Auction}

\hypersetup{
pdftitle={Sybil-proof diffusion auction in social networks},
pdfsubject={cs.GT},
pdfauthor={Hongyin Chen, Xiaotie Deng, Ying Wang, Yue Wu, Dengji Zhao},
pdfkeywords={First keyword, Second keyword, More},
}

\newcommand{\BibTeX}{\rm B\kern-.05em{\sc i\kern-.025em b}\kern-.08em\TeX}
\newtcolorbox{edited}[1][]{%
    colback=black!5,
    colframe=black!5,
    notitle,
    sharp corners,
    borderline west={2pt}{0pt}{red!80!black},
    enhanced,
    breakable,
    }

\newcommand{\vek}[1]{\boldsymbol{#1}}
\newcommand{\set}[1]{\mathrm{#1}}

\newcommand{\STM}{STM}
\newcommand{\STMfull}{Sybil tax mechanism}
\newcommand{\SCM}{SCM}
\newcommand{\SCMfull}{Sybil cluster mechanism}

\newcommand{\mechanism}{\mathcal{M}}
\newcommand{\maskprofile}{\vek\theta'}

\newcommand{\settype}[1]{\set\Theta_{#1}}
\newcommand{\setbeta}[2][]{#1{\beta}_{#2}}

\newcommand{\action}[2][]{#1{\vek\theta}_{#2}}

\newcommand{\eps}{\varepsilon}

\newcommand*{\defeq}{\stackrel{\text{def}}{=}}
\newcommand{\dom}{\operatorname{dom}}

\newcommand{\maxbid}[2][]{#1{\operatorname{Max}}[#2]}

\newcommand{\revenue}{\mathcal{R}}
\newcommand{\welfare}{\mathcal{SW}}

\newtheorem{observation}{Observation}
\newtheorem{definition}{Definition}
\newtheorem{theorem}{Theorem}
\newtheorem{lemma}{Lemma}

\makeatletter
\newtheorem*{rep@theorem}{\rep@title}
\newcommand{\newreptheorem}[2]{%
\newenvironment{rep#1}[1]{%
 \def\rep@title{#2~\ref{##1}}%
 \begin{rep@theorem}}%
 {\end{rep@theorem}}}
\makeatother
\newreptheorem{theorem}{Theorem}
\newreptheorem{lemma}{Lemma}
\newtheorem*{proof-sketch}{Proof Sketch}

\begin{document}

\maketitle
\begin{abstract}
A diffusion auction is a market to sell commodities over a social network, where the challenge is to incentivize existing buyers to invite their neighbors in the network to join the market. Existing mechanisms have been designed to solve the challenge in various settings, aiming at desirable properties such as non-deficiency, incentive compatibility and social welfare maximization. Since the mechanisms are employed in dynamic networks with ever-changing structures, buyers could easily generate fake nodes in the network to manipulate the mechanisms for their own benefits, which is commonly known as the Sybil attack. We observe that strategic agents may gain an unfair advantage in existing mechanisms through such attacks. To resist this potential attack, we propose two diffusion auction mechanisms, the Sybil tax mechanism (STM) and the Sybil cluster mechanism (SCM), to achieve both Sybil-proofness and incentive compatibility in the single-item setting. Our proposal provides the first mechanisms to protect the interests of buyers against Sybil attacks with a mild sacrifice of social welfare and revenue. 
\end{abstract}


\section{Introduction}\label{h1:intro}

Auction is an important method of selling commodities where the seller collects bids from buyers and allocates commodities according to these bids. Previous works~\cite{bulow1996auctions} have shown that more buyers would significantly lead to higher social welfare and revenue in auctions. However, buyers have no incentive to invite others to their auctions because it would cause tougher competition and
hurt their own interests. Recently, there has been an emergence of studies on the diffusion auction over social networks~\cite{guo:DM-survey}, which studies mechanisms that incentivize buyers to invite new agents to an auction via a social network. 
In these works, while the seller only knows her neighbors, any buyer who is informed of the auction may bid to buy, as well as diffuse the information about the sale to her neighbors to improve her utility.

The first work of diffusion auctions~\cite{li:IDM} proposed the information diffusion mechanism (IDM) for selling one item in a social network,
focusing on the incentive
compatibility of its information propagation action. 
Under
IDM, it's a dominant strategy for each bidder to truthfully bid her private valuation and to diffuse the auction information to all her neighbors.
\citet{zhao2019selling} and \citet{kawasaki2020strategy}
further designed diffusion mechanisms in selling multiple homogeneous items in a social network. 
Those works, however, did not consider a common threat 
known as 
the Sybil attack. 

The first study on Sybil attacks~\cite{douceur2002sybil} considered a situation in peer-to-peer systems where malicious agents may gain an unfair advantage by creating fake identities.
One such 
example is presented in
Figure~\ref{fig:sybil-identities} for a social network 
where the agent $x$ creates six false-name identities.
The Sybil attack is
a significant threat to auctions 
(commonly called \emph{false-name bids})
and has been extensively investigated in traditional auction settings~\cite{yokoo2004effect}. One example proven to be vulnerable to Sybil attacks 
is the well-known Vickrey-Clarke-Groves (VCG) auction for combinatorial auctions of at least two items.

Our work studies this fundamental issue in diffusion auctions on social networks, where fake nodes can be easily created. 
In existing mechanisms such as IDM and FPDM~\cite{zhang2021fixed}, intermediate buyers are rewarded for inviting more buyers. This makes Sybil attack highly profitable and harmful.

We consider a strong adversarial model that allows every buyer 
to create 
fake identities, which
can link with each other internally or link to the buyer's neighbors. But they don't have incoming arcs from other agents, since they are only visible to the creator. 
Our goal is to incentivize diffusion without encouraging Sybil attacks in diffusion auctions.

\subsection{Our Contribution}

We propose two Sybil-proof diffusion mechanisms, the Sybil tax mechanism (STM) and the Sybil cluster mechanism (SCM). STM achieves Sybil-proofness by identifying trustworthy agents. In STM, diffusing to ``suspicious'' agents is not beneficial to buyers.
However, agents in STM do not have a strong incentive to diffuse information.
In proposing SCM, we 
provide a stronger incentive for diffusion, where the reachability of each non-Sybil vertex is credited to some selected agents. By a mild sacrifice on
the seller's revenue, SCM creates a strong incentive to invite new
buyers.

Our work resolves several difficulties. 
We are the first to identify and model the Sybil attack in diffusion auctions.
Our adversarial model
is the most general form of Sybil attack possible without collusion.
Existing diffusion mechanisms cannot resist
such Sybil attacks, even after removing ``suspicious'' agents.

Additionally, we discuss the social welfare and revenue of Sybil-proof diffusion auctions. We prove that there is no optimal
SP diffusion auction mechanism for social welfare, and all SP mechanisms perform poorly in the worst case. To further exhibit the performance of STM and SCM, we conduct experiments under different settings. Experimental results indicate that STM and SCM do not sacrifice social welfare and revenue much compared with non-SP mechanisms.

\begin{figure}[t]
    \centering
    \includegraphics[width=0.6\linewidth]{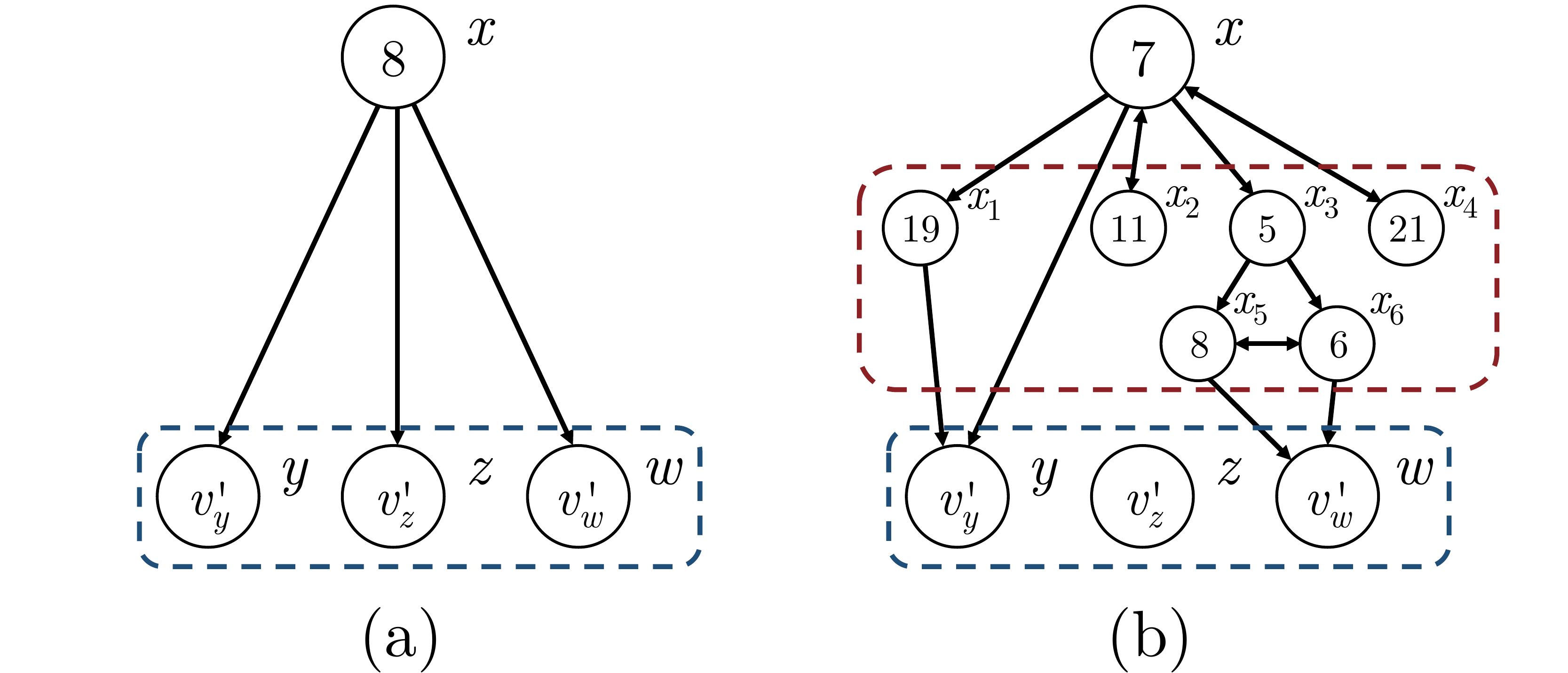}
    \caption{(a) the true type of \(x\); (b) a Sybil attack of \(x\) involving six Sybil identities (in the red rectangle).}
    \label{fig:sybil-identities}
\end{figure}

\subsection{Related Literature}
The Sybil attack has become a fundamental issue
in traditional social networks~\cite{viswanath2010analysis},
where nodes are usually divided into two types: honest ones and Sybil ones. 
In such settings, various protocols have been proposed to
identify Sybil nodes and maintain honest nodes through
the graph structures~\cite{yu2008sybillimit,liu2014defense}.

The Sybil attack is also destructive in auctions. The pioneering work to study
Sybil attacks on 
combinatorial auctions~\cite{yokoo2004effect}
proved that Sybil-proofness and Pareto optimality can't be achieved simultaneously. 
Many other works have followed. For example,
\citet{iwasaki2010worst}
have shown that a Sybil-proof combinatorial auction mechanism may result in extremely low social welfare in some cases.
In dynamic spectrum access auctions, Sybil-proofness has been only achieved when severe restrictions are imposed on Sybil agents. For example, PRAM~\cite{dong2021pram} requires that if an agent performs the Sybil attack, the sum of bids given by herself and her Sybil identities is equal to her private valuation.

\section{Preliminaries}\label{h1:preliminaries}

In a social network, a seller \(s\) is selling one item to a buyer among the set of potential buyers \(N=\{1,2,\dots,n\}\). The set \(N\) is unknown to the seller; instead, she only knows some buyers \(r(s)\subseteq N\). Likewise, each buyer \(i \in N\) has her private social connections, represented as a set of neighbors \(r(i)\subseteq N\). Each buyer $i \in N$ also has a private valuation of the item, which is denoted as $v_i$. Collectively, each buyer \(i\) owns a private type \(\theta_i=(v_i, r(i)) \in \mathbb{R}_{\geq 0}\times 2^N\).

In a diffusion auction, \(s\) can only advertise the sale to her neighbors \(r(s)\) initially. Then, each buyer \(i\) with the information of the sale may diffuse it to some of her neighbors in \(r(i)\). Recursively, many buyers can be informed.
Each buyer is asked to give a bid on the item besides diffusing the sale. The mechanism consequently sells the item to an informed buyer and rewards some buyers for their contribution of inviting others.

We model the bid and diffusion of buyer $i$ as the report type \(\theta'_i=(v'_i, r'(i))\in \mathbb{R}_{\geq 0}\times 2^{r(i)}\), where \(v'_i\in \mathbb{R}_{\geq 0}\) is her bid and \(r'(i)\) is the set of buyers she diffuses to.
A buyer can only diffuse to her neighbors (i.e. $r'(i) \subseteq r(i)$). The input of the mechanism is therefore a report profile \(\vek\theta'=(\theta'_1,\dots,\theta'_n)\), and we suppose the set of seller's neighbors \(r(s)\) is provided in advance and fixed.

The set of all possible types and reports of buyer $i$ is denoted as \(\set\Theta_i=\mathbb{R}_{\geq 0}\times 2^N\), and we denote the set of all possible profiles as \(\set\Theta\).

\begin{definition}[Diffusion auction mechanism]\label{def:diffusion-auction}
A diffusion auction mechanism \(\mathcal{M}\) is defined as a pair of allocation and payment schemes \((\vek\pi(\cdot),\vek{t}(\cdot))\) for arbitrary agent set \(N\):
\begin{itemize}
    \item allocation scheme \(\vek\pi : \set\Theta \to \{0,1\}^n\);
    \item payment scheme \(\vek t : \set\Theta \to {\mathbb{R}}^n\).
\end{itemize}
Given the reported type profile \(\vek\theta'=(\theta'_1,\dots,\theta'_n)\), whose length is not known in advance, \(\pi_i(\vek\theta')=1\) means that agent \(i\) wins the item, and \(0\) otherwise. She then pays \(t_i(\vek\theta')\in\mathbb{R}\) to the seller.
\end{definition}

We assume the following feasibility conditions for diffusion auctions throughout the paper: 

\begin{enumerate}
    \item Allocation feasibility: \(\sum_{i\in N}{\pi_i(\maskprofile)}\leq 1\),
    \item Anonymity: except for ties, the mechanism output is invariant to any permutation on \(N\), and
    \item \label{cond:ignorance} Ignorance of unreachable vertices: if \(i\in N\) is unreachable from \(s\) on the social network represented by \(\maskprofile\), then \(\pi_i(\maskprofile)=t_i(\maskprofile)=0\), and the mechanism output must be invariant with respect to \(\theta'_i\).
\end{enumerate}

We can use graph theory to formalize the third condition above. The social network represented by the true type profile \(\vek\theta\) can be denoted as a graph \(G\) with vertex set \(V(G)=\{s\}\cup N\) and directed edge set \(E(G)=\{(x,y)\mid x\in V(G), y\in r(x)\}\). Likewise, a graph \(G(\theta')\) can be defined for the report profile \(\maskprofile\). The subgraph of $G(\theta')$ with vertices that are reachable from $s$ is denoted as $G_{s}(\theta')$. All vertices unreachable from $s$ are excluded from it. When $\maskprofile$ can be inferred from context, we omit it and write \(G_{s}(\theta')\) as $G_s$.
The ignorance condition means that the mechanism can only use the structural information about \(G_s(\theta')\) and the bids of \(V(G_s)\) as inputs. This is a key difference between a diffusion mechanism and the traditional auction mechanism.

The agents have a quasi-linear utility model. Given a buyer's true type \(\theta_i=(v_i,r(i))\) and the report profile of all agents \(\maskprofile\), her utility under mechanism \(\mathcal{M}=(\vek\pi,\vek{t})\) is \(u_i(\theta_i, \maskprofile,\mathcal{M}) = v_i\cdot \pi_i(\maskprofile) - t_i(\maskprofile)\).
\subsection{Non-deficiency, Individually Rationality and Incentive Compatibility}\label{h2:ir-ic}

In this section, we define the objectives of diffusion mechanisms.

{
We say a mechanism is individually rational if any buyer can achieve a non-negative utility by reporting truthfully, no matter what the other agents do. This means that any agent is at least willing to participate.

\begin{definition}[IR]\label{def:IR}
A diffusion mechanism \(\mechanism=(\vek\pi,\vek{t})\) is \emph{ex-post individually rational} (IR) if for all \(\vek\theta\in\set\Theta\), for all \(i\in N\) with \(\theta_i=(v_i,r(i))\), it is guaranteed that \(u_i(\theta_i,\vek\theta,\mechanism)\geq 0\).\footnote{The definition of IR in previous literature in diffusion auctions does not require the buyer to truthfully diffuse, which differs from the traditional definition in AGT. In the setting of this paper, the two definitions are equivalent, and the traditional definition is presented.}
\end{definition}
}

A desired mechanism encourages agents to behave truthfully, i.e., to bid their private values and to diffuse the information to all their neighbors. In the diffusion auction setting, an agent may be strategic by overbidding, underbidding, or under-diffusion, if such strategies can bring her an advantage. Dominant-strategy incentive compatibility requires that reporting the true type is a dominant strategy for every buyer, ruling out these strategic reports.

\begin{definition}[DSIC]\label{def:IC}
A diffusion mechanism \(\mechanism=(\vek\pi,\vek{t})\) is \emph{dominant-strategy incentive compatible} (DSIC, or IC for short) if, for any buyer \(i\in N\) with type \(\theta_i\), any report profile of other agents \(\action{-i}'\) and any \(\theta'_i\in \set\Theta_i\) satisfying \(r'(i)\subseteq r(i)\), we have \(u_i(\theta_i, \vek\theta,\mechanism)\geq u_i(\theta_i, \maskprofile,\mechanism)\), where \(\vek\theta=(\theta_i,\action{-i}')\) and \(\maskprofile=(\theta'_i,\action{-i}')\).
\end{definition}

Some IC diffusion auction mechanisms, like VCG, may give a negative revenue to the seller~\cite{li:IDM}. We define the following non-deficiency condition to rule out these mechanisms. For a mechanism \(\mechanism=(\vek\pi,\vek{t})\) and a type profile \(\vek\theta\), the revenue to the seller is defined as \(\revenue^\mechanism(\vek\theta)=\sum_{i\in N}t_i(\vek\theta)\).

\begin{definition}\label{def:budget}
A diffusion mechanism \(\mechanism\) is \emph{non-deficit}, or weakly budget balanced, if its revenue for the seller is always non-negative, or formally, \(\revenue^\mechanism(\vek\theta)\geq 0\) for all \(\vek\theta\in\set\Theta\).
\end{definition}

\subsection{The Sybil Attack and Sybil-Proofness}\label{h2:sp}

In our setting, we further consider the desiderata of disincentivizing Sybil attacks. When a buyer \(i\) performs a Sybil attack, she creates an arbitrary number of Sybil identities (or false-name identities) \(i_1,i_2,\dots,i_k\), with a report profile of \(\theta'_{i_1},\dots,\theta'_{i_k}\). The set of all identities of \(i\) is denoted as \(\phi=\{i,i_1,\dots,i_k\}\). For every \(i_j\in \phi\), \(\theta'_{i_j}=(v'_{i_j},r'(i_j))\), it must be guaranteed that \(r'(i_j)\subseteq  r(i) \cup \phi\) because \(i\) does not know any agent besides herself, her neighbors, and her Sybil identities. 
Consequently, such Sybil identities of buyer \(i\) cannot have incoming edges from other buyers because the neighbor sets of other buyers cannot be changed by \(i\).
Refer to Figure~\ref{fig:sybil-identities} for an example of Sybil attacks.

We define Sybil-proofness as a criterion for ruling out such attacks. A mechanism is Sybil-proof if, for every buyer, any form of Sybil attack cannot bring a higher utility.

\begin{definition}[SP]
A diffusion mechanism \(\mechanism=(\vek\pi,\vek{t})\) is \emph{Sybil-proof} (SP) if, for any type profile \(\vek\theta\), any buyer \(i\in N\), and for all \(\theta'_{i},\theta'_{i_1},\dots,\theta'_{i_k}\in\settype{i}\) satisfying \(r'(i)\subseteq \phi\cup r(i)\) and \(\forall i_j\in \phi: r'(i_j)\subseteq \phi\cup r(i)\), we have 
\[u_i(\theta_i,\vek\theta,\mechanism)
\geq
u_i(\theta'_i,\vek\theta',\mechanism)+
\sum_{i_j\in \phi}u_{i_j}(\theta'_{i_j},\vek\theta',\mechanism)
\]
where the Sybil-attack report profile is \(\maskprofile=(\theta'_{i},\theta'_{i_1},\dots,\theta'_{i_k},\vek\theta_{-i})\).
\end{definition}

A Sybil attacker can bring an arbitrary number \(k\) of Sybil identities, and each of the \(k+1\) identities (including the agent herself) can report arbitrarily. This formulation is the most general form of Sybil attacks without collusion. As a degenerate case, a Sybil attack with \(k=0\) is equivalent to a single-agent strategic play in the previous diffusion action setting. Therefore, Sybil-proofness implies incentive compatibility.

\subsection{Vulnerability of Existing Mechanisms}\label{vul_existing}

To the best of our knowledge, none of the existing diffusion auctions is Sybil-proof except for the trivial Neighbor Second-Price Auction (NSP), where only the seller's neighbors are considered with a second price auction 
(see Appendix~\ref{app:idmvcg} for details).\footnote{There are confusingly two mechanisms named VCG in the literature of diffusion auctions: 1) the single-item VCG auction among the seller's neighbors, and 2) the generic VCG mechanism applied to diffusion auction~\cite{li:IDM}. To disambiguate, the former is called Neighbor Second-Price (NSP) in this paper.} As assumed, the seller's neighbors are known to the seller, so there is no chance for them to create fake identities to join NSP. 

Other existing mechanisms for diffusion auctions are all vulnerable to the Sybil attack. Here we use the two typical mechanisms proposed in \cite{li:IDM}, VCG and IDM, to demonstrate the possibility of Sybil attacks. 
Definitions of these mechanisms are given in Appendix~\ref{app:idmvcg}.

\begin{figure}[t]
    \centering
    \includegraphics[width=0.7\linewidth]{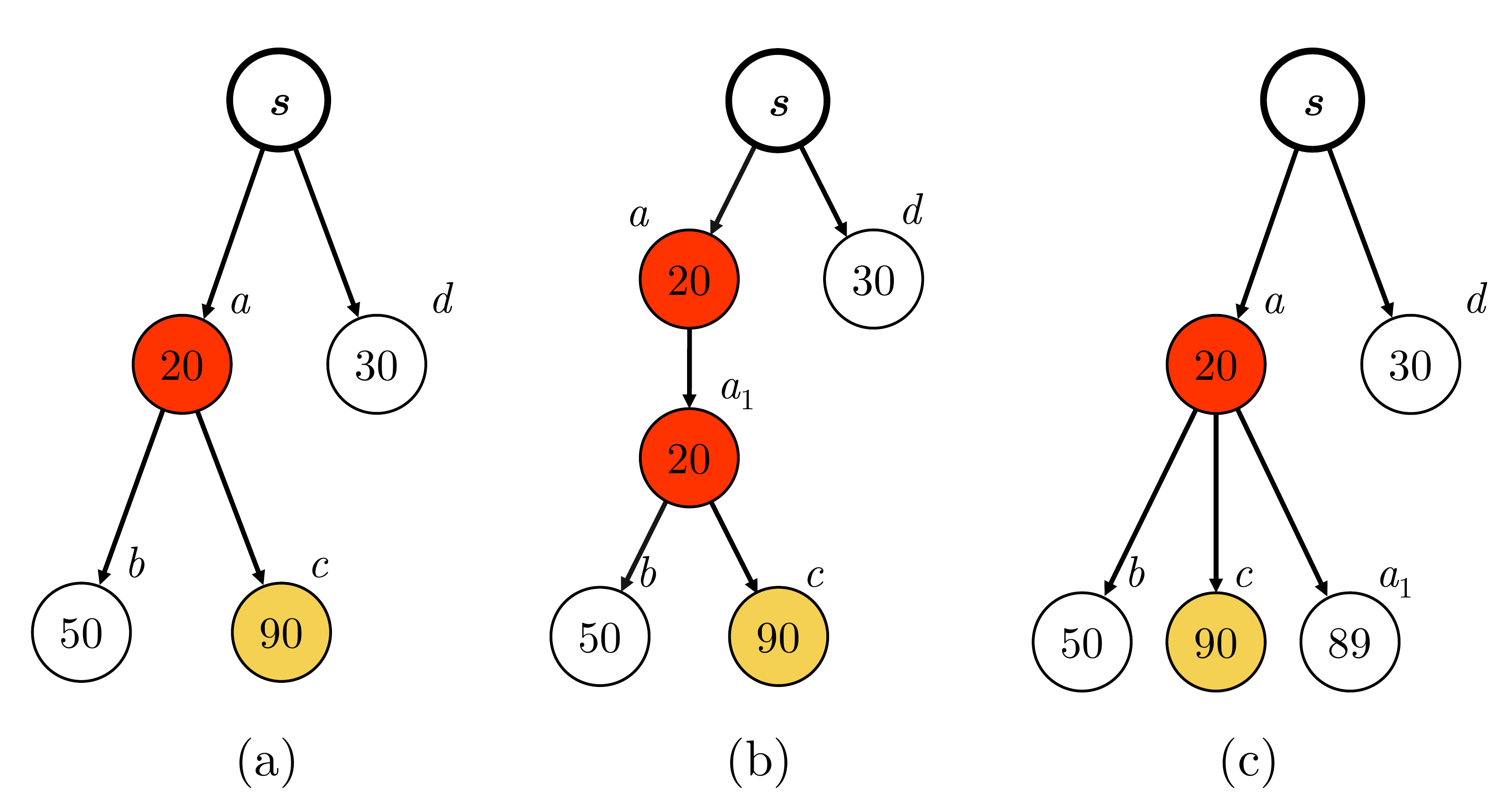}
    \caption{Sybil attack counterexamples of VCG and IDM.}
    \label{fig:nonSP}
\end{figure}

\begin{observation}\label{thm:nonSP-VCG}
VCG and IDM are not Sybil-proof.
\end{observation}

The classic VCG mechanism can be easily extended as a diffusion auction. Under VCG, the item is sold to the highest bidder, and other agents are paid the social welfare increase due to their participation. In the example shown in Figure~\ref{fig:nonSP}(a), if the intermediate node $a$ does not participate, $b$ and $c$ will be unable to join, and the social welfare will be $30$. With $a$'s participation, the social welfare is $90$, so VCG will pay $60$ to $a$. Now, if $a$ creates a fake identity $a_1$, then both $a$ and $a_1$ will be paid $60$ (a successful Sybil attack).

Since VCG paid a lot to the agents connecting the highest bidder to the seller, it cannot be non-deficit. Thus, IDM was proposed to guarantee that the seller's revenue is non-negative. IDM does not directly sell the item to the highest bidder; it uses a resale process to find the winner. It first allocates the item to the first cut point to reach the highest bidder, and the buyer pays the highest bid without her participation. In the example shown in Figure~\ref{fig:nonSP}(a), the item is first allocated to $a$ and $a$ pays $30$. Then $a$ can choose to resell it to $c$ and $c$ has to pay the highest without $c$ to $a$, which is $50$. Now, if $a$ creates a fake neighbor $a_1$ with bid $89$, then $c$ will need to pay $89$ to $a$ (another successful Sybil attack). 

We also proved that the other existing diffusion auction mechanisms~\cite{10.5555/3398761.3398947,zhang2021fixed} are not Sybil-proof in Appendix~\ref{app:attack}.

\section{Analysis of the Sybil Attack}\label{h1:graph-theory}

In this section, we study the features of Sybil attacks.
In our model, a Sybil identity \(y\) created by a real agent \(x\) can only be connected by her other Sybil identities or by \(x\) herself. This implies that every path from \(s\) to \(y\) contains \(x\).
In graph theory \cite{lengauer:dominators}, this is noted as \(x\) \emph{dominates} \(y\), or \(x \dom y\), and \(x\) is called a \emph{dominator} of \(y\). 
If a vertex $y$ has no dominator except $s$, one can be sure that $y$ is not a Sybil identity. Conversely, when vertex $y$ has a dominator $x \neq s$, there is a chance that $y$ is a Sybil identity of $x$.

In previous diffusion auction mechanisms like VCG and IDM, being a dominator can bring the agent profit (i.e., the mechanisms reward her for inviting new agents),
leaving room for one to profit from Sybil attacks. This explains why Sybil-proofness is hard to achieve in diffusion auctions.

An \emph{immediate dominator} of \(x\), denoted as \(v=idom(x)\), is defined as the unique vertex \(v\) who dominates \(x\) and is dominated by every other dominator \(w\neq v\) of \(x\).

\begin{theorem}\label{thm:dom-tree}
Every vertex on the graph \(G_s(\theta')\) except \(s\) has an immediate dominator, and the edges \(\{(idom(x),x)\mid x\in N\}\) form a directed tree with \(s\) being its root, called the \emph{dominator tree} of \(G_s(\theta')\) rooted at \(s\).
\end{theorem}

This is exactly Theorem~1 of \cite{lengauer:dominators}. The definition of dominators is identical to \emph{diffusion critical nodes} in \cite{li:IDM}, and the path from \(s\) to \(x\) on the dominator tree is the \emph{diffusion critical sequence} of \(x\).

\subsection{Graphical Non-Sybil Agents}\label{h2:gamma}

In this subsection, we use graph theory to characterize the set of vertices that cannot be Sybil identities.
Firstly, the seller and her neighbors are not Sybil identities. In a real-world scenario, there are sometimes trustworthy entities like public figures and centralized institutions. Thus, we introduce an optional set of vertices \(\Gamma_0\), which is provided externally and guaranteed not to contain any Sybil identities. If no such vertices are provided, \(\Gamma_0=\emptyset\). Allowing such external information makes our mechanisms more flexible.

For the convenience of expression, we first give the definition of meeting points.

\begin{definition}[Meeting points]
For a pair of vertices \(x, y\), a vertex \(z\) is defined to be a \emph{meeting point} of \(x\) and \(y\) if there are two vertex-disjoint paths to \(z\), from \(x\) and \(y\) respectively.
\end{definition}

If a vertex is a meeting point of two other non-Sybil vertices \(x, y\), it must not be a Sybil identity. This is because all paths from non-Sybil vertices to a Sybil identity $i_j$ must contain its owner $i$ which contradicts the definition of meeting points. Therefore, we have the following definition of graphical non-Sybil agents which iteratively collects meeting points of existing members.

\begin{definition}[Graphical non-Sybil agents]
The set \(\Gamma(\theta') \subseteq V(G_s)\) is defined as follows:
\begin{enumerate}
    \item Initialize the set as \(\Gamma(\theta') := \{s\} \cup r(s) \cup \Gamma_0\).
    \item \label{rule:trusted-join} For each pair of vertices \(x,y\in \Gamma(\theta') \), if \(z\) is a meeting point of them in graph $G_s$, then add \(z\) to the set, i.e. \(\Gamma(\theta') :=\Gamma(\theta') \cup \{z\}\).
    \item Repeat step \ref{rule:trusted-join} until there are no more vertices to add.
\end{enumerate}
\end{definition}

It can be shown that
\(\Gamma(\theta')\) is precisely the maximal set of vertices that cannot be Sybil identities. This will be proven in Lemma~\ref{lem:maximal-non-sybil} after the introduction of Sybil clusters.

\subsection{Overly Sensitive Mechanism}\label{h2:osm}

Given the graphical non-Sybil agents, a straightforward idea to achieve Sybil-proofness is to apply the existing diffusion mechanisms on non-Sybil agents.
This idea of detection and removal is a common solution to Sybil attacks in social networks~\cite{liu2014defense,wang2020structurebased,xu2010resisting}. However, we find that such an approach doesn't work because an agent can misreport her neighbor set and turn non-Sybil agents into suspicious ones.

We propose the \emph{overly sensitive} mechanism (OSM) to show why such an idea does not work. In OSM, we ignore all potential Sybil identities (i.e. all $i \notin \Gamma(\maskprofile)$) and focus on the reachable part of the induced subgraph inducing from $\Gamma(\maskprofile)$, denoted as \(G_s[\Gamma(\maskprofile)]\). The subgraph $G_s[\Gamma(\maskprofile)]$ contains only vertices in $\Gamma(\maskprofile)$, and for each vertex $x$ in it, there is a path from $s$ to $x$ that only passing non-Sybil agents. We adopt IDM on \(G_s[\Gamma(\maskprofile)]\).

OSM seems Sybil-proof because Sybil identities are all ruled out. However, we find that OSM is not even incentive compatible. In OSM, the detection-and-removal process can be exploited by malicious agents. In Figure \ref{fig:overly-sensitive-counterexample}a, every vertex will be in \(\Gamma(\vek\theta)\). Under IDM, \(a\) will buy the item with the second-highest price \(v'_c=9\). However, if \(a\) chooses not to diffuse the information to \(c\) as in Figure \ref{fig:overly-sensitive-counterexample}b, $c$ would be excluded from \(G_s[\Gamma(\maskprofile)]\), and $a$ would get the item with a lower payment of \(7\).

Since SP implies IC, OSM is not Sybil-proof either. Therefore, we need a new approach to resist Sybil attacks in diffusion auctions.

\begin{figure}[thb]
    \centering
    \includegraphics[width=0.7\linewidth]{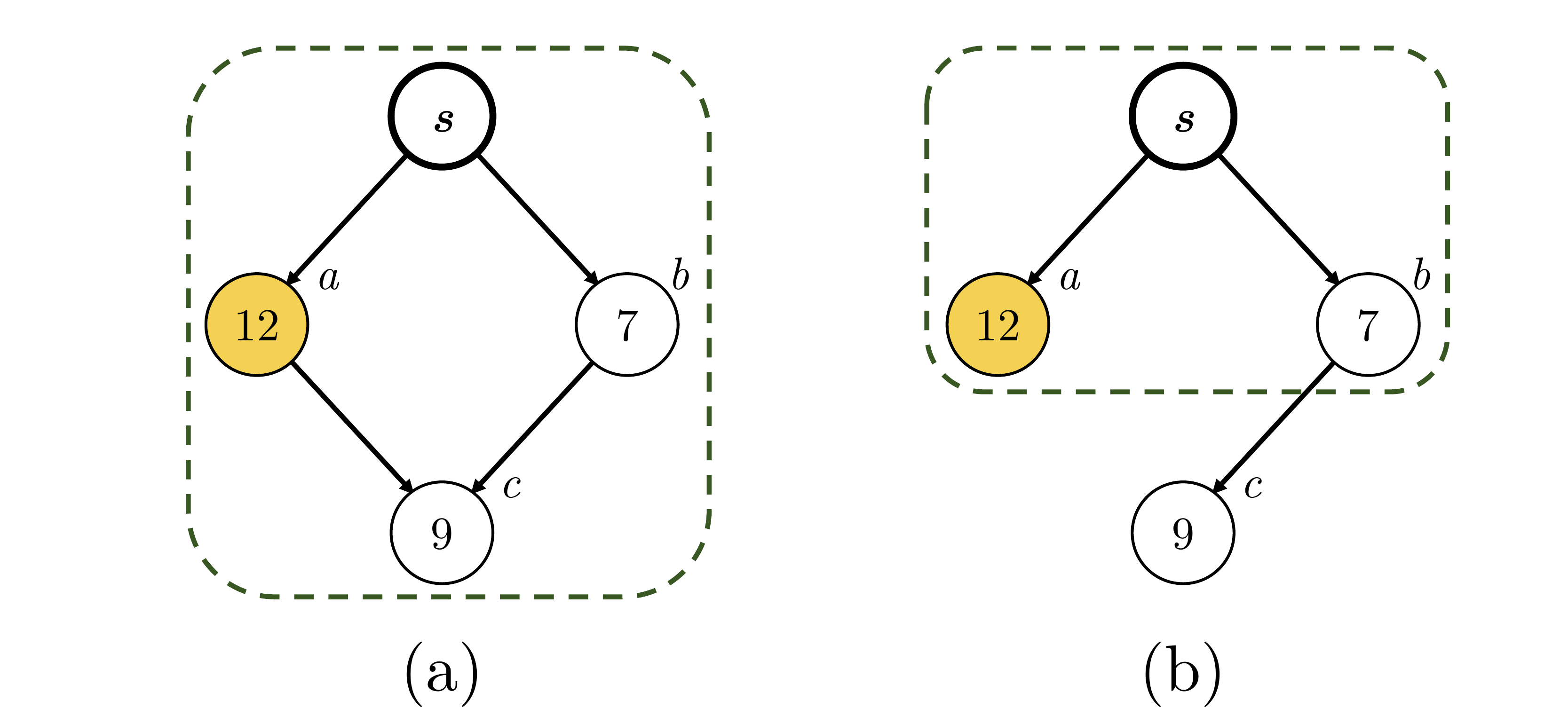}
    \caption{A counterexample for the overly sensitive mechanism. The set \(\Gamma(\maskprofile)\) is denoted by the dashed border rectangle.}
    \label{fig:overly-sensitive-counterexample}
\end{figure}

\section{Sybil Tax Mechanism}\label{h1:stm}

In this section, we present the first main contribution of this paper, our first Sybil-proof diffusion mechanism, called Sybil Tax Mechanism (STM).

Before describing it, we introduce some notations. We use \(\maxbid{S}\) to denote the highest bid in a set \(S\), that is,
\(\maxbid{S}=\max_{x\in S} v'_x\).
We also denote the vertices she dominates as \(\alpha(x)=\{y\mid x\dom y\}\) for every vertex \(x\). The vertex \(x\) is critical for these \(y\in \alpha(x)\), because without her diffusion, these vertices are not reachable from $s$. It is also known as \emph{diffusion critical children} in the terminology of previous literature on diffusion auctions.

\begin{framed}
 \noindent\textbf{Sybil Tax Mechanism (STM)}\\
 \rule{\textwidth}{0.5pt}

\begin{enumerate}
    \item Given the reported type profile \(\maskprofile\), we first calculate the reachable reported graph $G_s(\maskprofile)$ and the graphical non-Sybil agent set $\Gamma(\maskprofile)$. Let's write them as $G_s$ and $\Gamma$ for short.
    
    \item Find the reachable buyer with the highest bid, denoted by \(x^*\), where \(v'_{x^*}=\maxbid{V(G_s)}\).
    
    \item Compute the dominator sequence \(C_{x^*}=\{c_0=s,c_1,\dots,c_\ell=x^*\}\). Specifically, we have \(c_j = idom(c_{j+1})\), for all \(0 \le j < \ell\).

    \item We define \(p_j\), the buying price of $c_j$,
    as the highest bid without the participation of \(c_j\). The selling price of $c_j$, denoted as \(q_j\), is defined as the highest bid among all vertices that are identified as not the Sybil identity of \(c_j\). Formally,
\begin{align*}
    p_j &= \maxbid{V(G_s)\setminus \alpha(c_j)}&\text{for }1 < j \le \ell,\\
    q_j &= \maxbid{(V(G_s)\setminus \alpha(c_j))\cup \setbeta{j}} & \text{for }1 \le j < \ell,
\end{align*}
where 
$$\setbeta{j}=\{x \mid \exists{y\neq c_j}, y\in (\alpha(c_j) \setminus \alpha({c_{j+1}})) \cap \Gamma, x \in \alpha(y)\}.$$

\item Pick a \(c_d\) with the lowest index \(d\) that satisfies \(v'_{c_d}\geq q_d\). If such $d$ does not exist among index $1 \le d < \ell$, we set $d = \ell$. This agent $c_d$ wins the item with \(\pi_{c_d}(\maskprofile)=1\). The payment function is calculated as 
\[
t_{c_j}(\maskprofile)=\left\{
\begin{aligned}
p_j &{} - q_j &\text{for } & 1 \le j < d,\\
p_j &         &\text{for } & j = d.
\end{aligned}
\right.
\]

\item The payment and allocation of all other buyers are zero.

\end{enumerate}
\end{framed}

In this mechanism, the item is sold along the dominator sequence from $s$ to $x^*$ as a series of successive transactions between neighboring agents. Agent \(j\)'s buying price $p_j$ is set as other agents' optimal social welfare (i.e. the highest bid of them) when she does not participate in the auction. This ensures that her report cannot lower her buying price. The agent \(j\) can sell the item further down the critical sequence to reach more potential buyers with a selling price of $q_j$. To achieve Sybil-proofness, we need to Sybil-attacking to be not profitable, i.e. not able to increase \(q_j\). When the item is passed from \(c_j\) to \(c_{j+1}\), since the latter may be a Sybil identity of the former, the selling price \(q_j\) of \(c_j\) must be irreverent to the report of \(c_{j+1}\). Indeed, \(q_j\) is defined as the highest bid among those who are guaranteed not the Sybil identity of her.
The set $\setbeta{j}$ is defined in a way that it is monotonically increasing with the report of \(c_j\) to incentivize diffusion, and that it contains no Sybil identity.

Conceptually, a buyer who gets the item can choose to keep it or to resell. She will pass the item only when her selling price is higher than her private value. STM simulates this choice based on buyers' bid through the choice of the winner \(d\).

This series of transactions are summed up by STM. In a single transaction, buyer \(c_j\) will receive \(q_j\) units of money and \(c_{j+1}\) pays \(p_{j+1}\) for it.
The price difference \(p_{j+1}-q_j\) can be considered as a ``tax'' paid by the intermediate buyers (which we call \emph{brokers}) to prove their innocence.

Since \(V(G_s)\setminus \alpha(c_j)\subseteq (V(G_s)\setminus \alpha(c_j))\cup \setbeta{j} \subseteq V(G_s)\setminus \alpha(c_{j+1})\), we have \(p_j\leq q_j\leq p_{j+1}\), so the monetary gain of brokers and the tax are all non-negative. This leads to individual rationality and non-deficiency.

\begin{figure}[tb]
    \centering
    \includegraphics[width=0.7\linewidth]{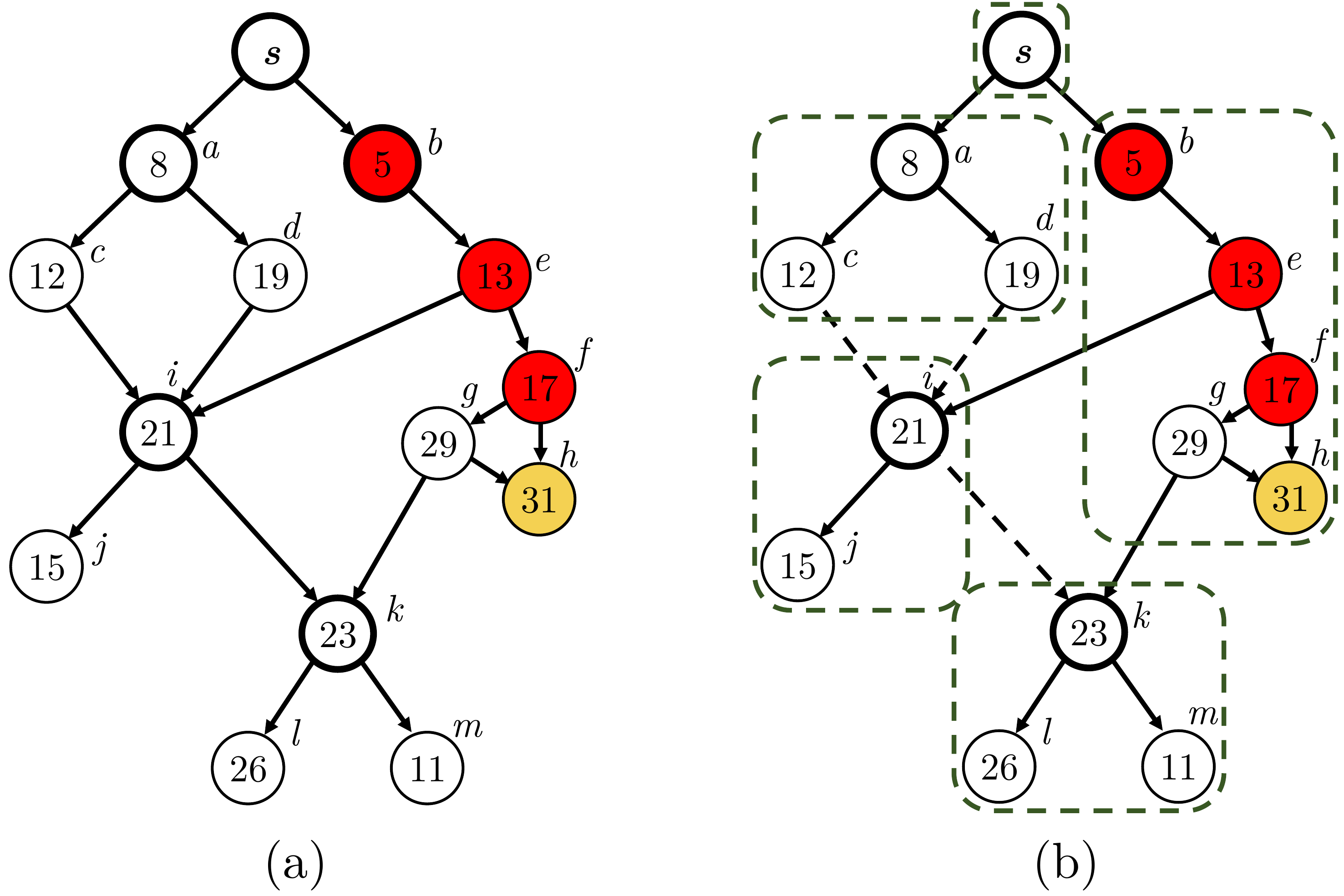}
    \caption{Examples of the {\STMfull} and the \SCMfull.}
    \label{fig:runningExample}
\end{figure}

Figure \ref{fig:runningExample}(a) illustrates {\STM} with an example. We assume that the externally provided set \(\Gamma_0\) is empty. When all buyers report their type truthfully, the mechanism runs as follows.

The set of graphical non-Sybil agents \(\Gamma(\vek\theta)\) is calculated as $\{s,a,b,\allowbreak{}i,k\}$. The mechanism identifies the highest-bidder $h$ and calculates the dominator sequence $C_{h} = \{c_0 = s, c_1 = b, c_2 = e, c_3 = f, c_4 = h\}$. The item is sold to $c_4=h$ because $c_4$ is the only buyer on the dominator sequence that satisfies \(v'_{c_j}\geq q_j\). Then we calculate the payments. For brokers \(c_1, c_2, c_3\), \(p_1 = p_2 = p_3 = 26 = q_1 = q_2 = q_3\), so they get paid \(0\). The winner \(c_4 = h\) pays \(p_4 = 29\). The seller gets a revenue of \(29\). 
In short, the buyer $h$ will pay $29$ to buy the item. Other buyers get zero utility.

\begin{theorem}[Main]\label{thm:main}
{\STM} is IR, non-deficit and Sybil-proof.
\end{theorem}
The proof of Theorem~\ref{thm:main} will be elaborated in Appendix~\ref{app:stm}. We provide a proof scketch here.

Intuitively, {\STM} is individually rational because \(q_j\geq p_j\) and \(v'_{c_d}\geq p_d\). A non-Sybil-attacking buyer \(c_j\) would want to maximize \(\beta_j\) to maximize her utility, which can be achieved by maximal diffusion. By the graph-theoretic properties of Sybil attacks, if a Sybil attack happens on the dominator sequence, the identities of the same buyer must be contiguous on the sequence, and the tax paid by such brokers would disincentivize this attack.

The above theorem shows that STM is incentive compatible because SP implies IC. In a previous work \cite{li:CDM}, \citeauthor{li:CDM} identified one class of diffusion mechanisms called critical diffusion mechanism (CDM) on social graphs, which covers a large class of incentive compatible mechanisms. The successive reselling in STM resembles CDM, but STM is not a member of that class. By introducing non-Sybil agents externally (i.e. $\Gamma_0\not=\emptyset$), STM can contribute the occurrence of some ``isolated'' non-Sybil agents to buyers in the dominator sequence.

Recall the example in Figure~\ref{fig:runningExample}, and we can see that every buyer other than the item's winner has zero utility. The following lemma shows that this is not a fluke. In essence, all possible profits of the brokers are taxed by the seller.
The proof can also be found in Appendix~\ref{app:stm}.
\begin{lemma}\label{lem:0spdm-zero}
In \emph{{\STM}} when \(\Gamma_0=\emptyset\), every buyer, except the item winner, has a payment of zero, and thus zero utility.
\end{lemma}

\section{Sybil Cluster Mechanism}\label{h1:scm}
In STM, we reward the brokers for their contribution to introducing agents in \(\Gamma(\maskprofile)\). However, when \(\Gamma_0\) is empty, no one other than the seller can bring a graphical non-Sybil agent on her own. This leads to zero profit for brokers, as shown in Lemma~\ref{lem:0spdm-zero}. As a result, their interests would be neither increased nor decreased through diffusions. Therefore, although STM is incentive-compatible, buyers' incentive to invite other agents is weak.

In this section, we create a positive incentive for inviting without losing Sybil-proofness. By removing some edges from the reachable reported graph $G_s(\maskprofile)$ while keeping \(\Gamma(\maskprofile)\) unchanged, we attribute the introduction of non-Sybil agents to some brokers and reward them. We will introduce a clustering process to accomplish this.

\subsection{Clustering}\label{h2:clustering}

\begin{definition}[Sybil clusters]
For every \(x\in \Gamma(\maskprofile)\), we define its \emph{Sybil cluster} \(K_x\) as below:

The cluster \(K_x\) contains vertex \(t\) if and only if there is a path from \(x\) to \(t\) on $G_s(\maskprofile)$ that does not contain any vertex in \(\Gamma(\maskprofile)\) other than \(x\) itself.

\end{definition}

The Sybil cluster $K_x$ includes $x$ and all vertices that are suspected of being the Sybil identities of $x$. Call the vertex \(x\) the \emph{root} of \(K_x\), who is the only member of $K_x$ that is also in the non-Sybil set $\Gamma(\maskprofile)$. The clusters $\{K_x \mid x \in \Gamma(\maskprofile)\}$ get the name because they form a partition of \(V(G_s(\maskprofile))\).

\begin{lemma}\label{lem:cluster-disjoint}
Sybil clusters are disjoint, and every vertex \(t\) in \(V(G_s(\maskprofile))\) belongs to some Sybil cluster \(K_x\).
\end{lemma}
The proof can be found in Appendix~\ref{app:clustering}. 

Using Sybil clusters, one can
prove that \(\Gamma(\vek\theta')\) is the maximal set of guaranteed non-Sybil vertices.
\begin{lemma}\label{lem:maximal-non-sybil}
Any vertex \(t\) in \(G_s(\vek\theta')\setminus\Gamma(\vek\theta')\) may be a Sybil identity of some other vertex in \(\Gamma(\vek\theta')\).
\end{lemma}
\begin{proof}
Given a report profile \(\vek\theta'\), we can compute \(\Gamma(\vek\theta')\) and the Sybil clusters by definition. 
For any \(t\not\in \Gamma(\vek\theta')\), there exists \(x\in \Gamma(\vek\theta')\) such that \(t\in K_x\) from Lemma~\ref{lem:cluster-disjoint}.
Let \(\tilde{r}(x)=\bigcup_{u\in K_x}{r(u)}\setminus K_x\), and \(\tilde\theta_x=(\tilde{r}(x), v'_x)\). 
We can see that, under the true type profile \(\vek{\tilde\theta}=(\tilde\theta_x, \vek\theta'_{-K_x})\), the agent \(x\) may create Sybil identities \(\phi=K_x\) and make the report profile identical to \(\vek\theta'\). This shows that \(t\) may be a Sybil identity of \(x\).
\end{proof}

\subsection{SCM Mechanism}\label{h2:scm}

\begin{framed}
 \noindent\textbf{Sybil Cluster Mechanism (SCM)}\\
 \rule{\textwidth}{0.5pt}

\begin{enumerate}
    \item Given the reported type profile \(\maskprofile\) as input, we reconstruct a social network graph $H$ with vertices in $\Gamma(\maskprofile)$. Formally, \(H=(\Gamma(\maskprofile),E(H))\), where
\[
E(H)=\left\{\left(x,y\right)\mid \exists{i\in K_x, j \in K_y} \text{ s.t. }(i,j)\in E(G_s(\maskprofile))\right\}.
\]

    \item Sample a random shortest-path tree{\footnotemark} of \(H\) with equal probability and denote it as \(T_H\).  

    \item We construct a subgraph \(\hat{G}\) of \(G_s(\maskprofile)\) using \(T_H\). 
    
    Formally, \(\hat{G}=(V(G_s(\maskprofile)),E(\hat{G}))\) where $E(\hat{G})$ is defined as
$$ E(G_s(\maskprofile)) \setminus \{(i,j) \mid 
i \in K_x, j \in K_y, (x,y) \in E(H) \setminus E(T_H)\}.$$
Specifically, edge \((i,j)\) on graph $G_s(\maskprofile)$ is deleted if \(i \in K_x, j \in K_y, x \not= y\) and \((x,y)\not\in T_H\). All the remaining edges form a new graph \(\hat{G}\).

    \item Perform {\STM} with \(G_s = \hat{G}\), \(\Gamma = \Gamma(\maskprofile)\) rather than \(G_s = G_s(\maskprofile)\), \(\Gamma = \Gamma(\maskprofile)\) on the agents' reports.
\end{enumerate}
\end{framed}
\footnotetext{For every vertex \(x\in V(H)\), we denote the shortest-path length from \(s\) to it on graph $H$ as \(dis_x(H)\). A spanning tree \(T_H\) of \(H\) is a subgraph of $H$ with $V(T_H) = V(H)$, which is also a directed tree. A spanning tree $T_H$ is said to be a shortest-path tree if, for every vertex \(x\in V(H)\), \(dis_x(T_H) = dis_x(H)\). A uniformly distributed random shortest-path tree can be generated by independently selecting a parent \(y\) for each \(x\neq s\), where $y$ is selected from $\{y\mid dis_x(H)=dis_y(H) +1, (x,y) \in E(H)\}$ with equal probability.}

In SCM, we remove some edges in $G_s(\maskprofile)$ according to the randomly selected shortest-path tree $T_H$ and keep $\Gamma(\maskprofile)$ as graphical non-Sybil agents. The appearance of some vertices in $\Gamma$ can be attributed to some brokers, thus increasing their profit.

SCM is also Sybil-proof and individually rational. We will provide a proof sketch here; rigorous proof can be found in Appendix~\ref{app:scm}.

\begin{theorem}\label{thm:scm-main}
    {\SCMfull} is IR, non-deficit, and Sybil-proof.
\end{theorem}

Individual rationality and non-deficiency follow trivially from the fact that {\STM} is IR and non-deficit. 
By the selection of the shortest-path tree, the diffusion choice of a buyer can affect vertices on the tree whose distance from $s$ is higher than her distance from $s$. Maximally diffusing for a buyer would bring her a more favorable tree structure and give her a better income.
Moreover, we find that Sybil attacks are completely ineffective in the clustering process. Combined with the Sybil-proofness of {\STM}, we can show that {\SCM} is Sybil-proof.

\begin{figure}[t]
    \centering
        \includegraphics[width=0.825\linewidth]{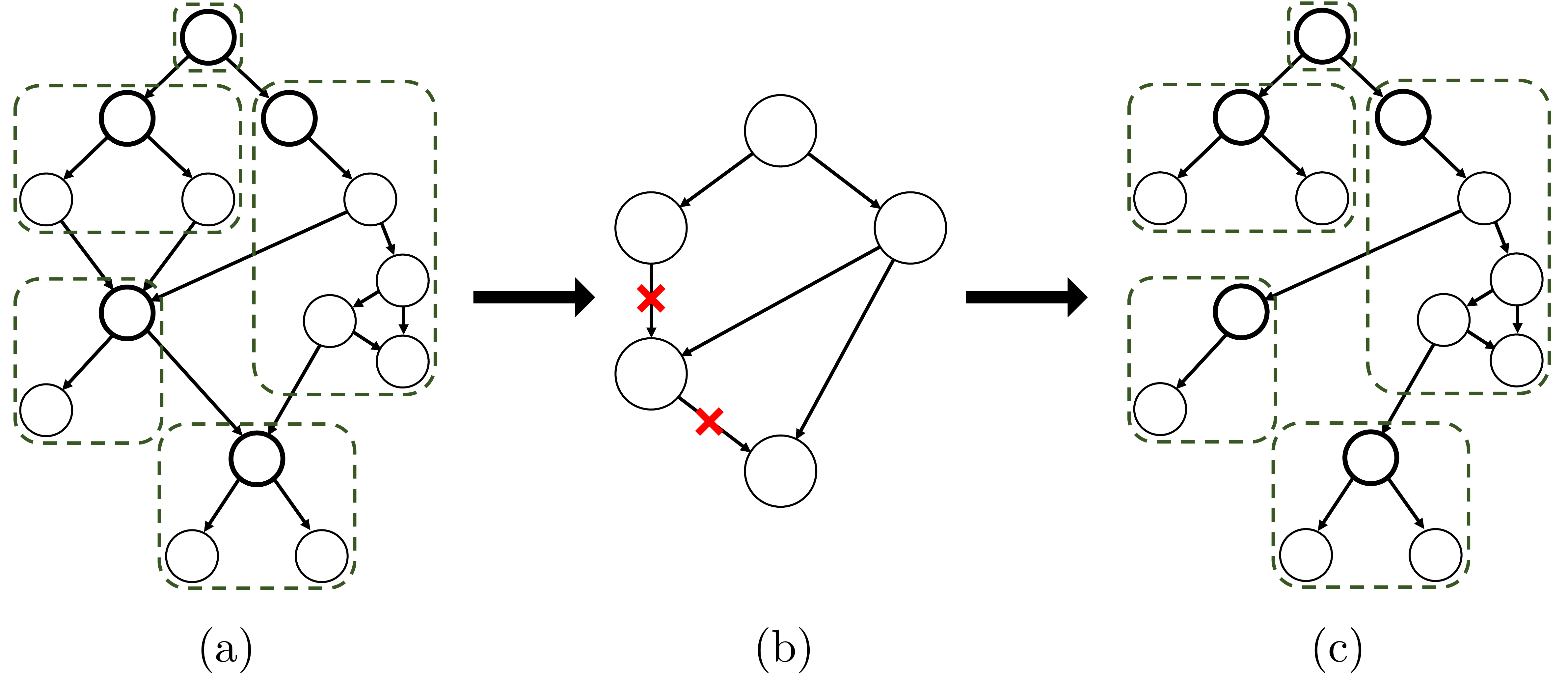}
    \caption{A visualization of SCM.}
    \label{fig:clusters}
\end{figure}

The example in Figure~\ref{fig:runningExample}b shows an example of the \SCMfull. The clustering process and a possible edge-removing process are shown in Figure~\ref{fig:clusters}. 
Assuming that all buyers report their true type, {\SCM} runs as follows:

The mechanism divides \(V(G_s(\vek\theta))\) into five Sybil clusters \(\{K_s, K_a, K_b, \allowbreak K_i, K_k\}\), where \(K_s = \{s\}, K_a = \{a, c, d\}, K_b = \{b, e, f, g, h\}, K_i = \{i, j\},\) and \(K_k = \{k, l, m\}\).
The mechanism randomly picks a shortest-path tree \(T_H\) and constructs a subgraph \(\hat{G}\).
We only show the case when the mechanism picks tree $T_H$ as Figure~\ref{fig:clusters}(b), where the mechanism deletes edge \((a, i)\) and edge \((i, k)\) in graph \(H\). In this case, edges \((c, i)\), \((d,i)\) and \((i, k)\) are removed from \(G_s(\vek\theta)\).
With \(\Gamma = \{s, a, b, i, k\}\), we perform {\STM} on \(\hat{G}\).

{\STM} identifies the buyer with the highest bidder to be $h$ and calculates the dominator sequence \(C_{h} = \{c_0 = s, c_1 = b, c_2 = e, c_3 = f, c_4 = h\}\). Since \(q_1 = 19 > v'_{c_1} = 5\), \(q_2 = 21 > v'_{c_2} = 13\) and \(q_3 = 26 > v'_{c_3} = 17\), we select \(c_4 = h\) as the winner of the item.

For the payments, broker \(c_1 = b\) pays \(t_{b} = p_1 - q_1 = 19 - 19 = 0\), \(c_2 = e\) gets \(-t_{e} = q_2 - p_2 =21 - 19= 2\) units of money, and \(c_3=f\) gets \(-t_{f} = 5\). The winner \(c_4 = h\) pays \(t_{h} = p_4 = 29\). The seller gets a revenue of \(22\).

\section{Discussion}\label{h1:discussion}

In this paper, we propose two Sybil-proof mechanisms, {\STM} and {\SCM}. In this section, we evaluate their performance on social welfare and revenue. Comparing our mechanism with the non-diffusion mechanism (i.e., NSP), other potential SP mechanisms and existing diffusion mechanisms (e.g., IDM, VCG) which are not SP, we raise three key questions.
\begin{enumerate}
    \item Do our diffusion mechanisms have better performance than non-diffusion ones?
    \item Does {\STM} or {\SCM} achieve optimal social welfare and revenue among all SP mechanisms?
    \item Compared with existing diffusion mechanisms, how much do our mechanisms sacrifice to achieve Sybil-proofness?
\end{enumerate}

We conduct theoretical and experimental analysis to answer these questions. For the first question, we prove that our mechanisms always achieve higher (or equal) social welfare and revenue than NSP. Our experimental results indicate that advantages of {\STM} and {\SCM} are significant. For the second question, we conduct worst-case analysis and show that every SP mechanism has extremely lower social welfare and revenue than another in some cases. Therefore, there is no optimal SP mechanism 
in terms of worst-case performance.
For the last question, our experimental results show that STM and SCM do not sacrifice social welfare and revenue much compared with non-SP mechanisms.
To eliminate the external effect, we assume that \(\Gamma_0 = \emptyset\) in this section. 

\subsection{Comparison}\label{h2:comparison}

We use $\welfare^\mechanism(\vek\theta)$ and  \(\revenue^{\mechanism}(\vek\theta)\) to denote the social welfare and revenue of the mechanism \(\mechanism\) under $\vek\theta$ respectively. We have
\[\welfare^\mechanism(\vek\theta)\defeq\sum_{x\in N}{\pi_x^{\mechanism}(\vek\theta)\cdot v_x}.\]
Recall that we have defined \(\revenue(\vek\theta) = \sum_{i\in N}t_i(\vek\theta)\) in Section~\ref{h2:ir-ic}.

The following theorem
shows that
both of our mechanisms outperforms the non-diffusion NSP mechanism. Under our mechanisms, agents' invitations indeed benefit the seller and the society.

\begin{theorem}\label{thm:compare_STM&SCM}
For all possible type profile \(\vek\theta\), we have
\begin{gather*}
\revenue^\mathrm{STM}(\vek\theta) \ge \revenue^\mathrm{SCM}(\vek\theta) \ge \revenue^\mathrm{NSP}(\vek\theta),\\
\welfare^\mathrm{STM}(\vek\theta) \ge \welfare^\mathrm{SCM}(\vek\theta) \ge \welfare^\mathrm{NSP}(\vek\theta).
\end{gather*}

\end{theorem}

We are curious whether STM achieves higher social welfare and revenue than all SP mechanisms. However, we'll show in Section~\ref{noawea} that none of SP mechanisms always has optimal social welfare and revenue.

The following theorems qualitatively examine the cost of Sybil-proofness. In Theorem~\ref{thm:revenue_compare_STM}, we find that STM achieve better revenue than the most cited diffusion auction, IDM \cite{li:IDM}. However, social welfare is sacrificed to achieve Sybil-proofness.
Theorem~\ref{thm:no_clear} reflects that there is no clear-cut comparison of the seller's revenue
between SCM and IDM, or between SCM and VCG.

\begin{theorem}\label{thm:revenue_compare_STM}
For any possible type profile \(\vek\theta\), we have
\begin{gather*}
\revenue^\mathrm{STM}(\vek\theta) \ge \revenue^\mathrm{IDM}(\vek\theta) \ge \revenue^\mathrm{VCG}(\vek\theta)\\
\welfare^\mathrm{SCM}(\vek\theta) \le \welfare^\mathrm{STM}(\vek\theta) \le \welfare^\mathrm{IDM}(\vek\theta) \le \welfare^\mathrm{VCG}(\vek\theta).
\end{gather*}
\end{theorem}

\begin{theorem}\label{thm:no_clear}
There exist two report profiles $\vek\theta_1, \vek\theta_2$, such that
\begin{gather*}
\revenue^\mathrm{SCM}(\vek\theta_1) > \revenue^\mathrm{IDM}(\vek\theta_1), 
\revenue^\mathrm{SCM}(\vek\theta_1) > \revenue^\mathrm{VCG}(\vek\theta_1), \\ \revenue^\mathrm{SCM}(\vek\theta_2) < \revenue^\mathrm{IDM}(\vek\theta_2), \revenue^\mathrm{SCM}(\vek\theta_2) < \revenue^\mathrm{VCG}(\vek\theta_2).
\end{gather*}

\end{theorem}

The proofs of Theorem~\ref{thm:compare_STM&SCM}, \ref{thm:revenue_compare_STM}, and  \ref{thm:no_clear} can be found in Appendix~\ref{app:compare}.

\subsection{Worst-Case Efficiency Analysis and (No) Optimality}\label{noawea}

In this subsection, we conduct worst-case analysis on SP mechanisms to explore the optimality of social welfare and revenue.
We consider the concept of worst-case efficiency ratio, which is adopted from previous work \cite{iwasaki2010worst} to measure the social welfare of Sybil-proof combinatorial auctions in the worst case. The worst-case efficiency ratio of $\mechanism$ indicates the ratio of $\mechanism$'s social welfare and the optimal social welfare in the worst-case input.

\begin{definition}\label{def:wc-eff-ratio}
    Given a type profile \(\vek\theta\), the optimal social welfare \(\welfare^*(\vek\theta)\) is defined to be the highest private value \(\max_{x\in V(G_s)}{v_x}\).
    The worst-case efficiency ratio of a mechanism $\mechanism$ is defined as follows:
    \[\inf_{\vek\theta}\frac{\welfare^\mechanism(\vek\theta)}{\welfare^*(\vek\theta)}.\]
\end{definition}

\begin{theorem}\label{thm:sp-eff-ratio}
    The worst-case efficiency ratio of any non-deficit, IR, and Sybil-proof diffusion auction mechanism is zero.
\end{theorem}

The above theorem shows that the social welfare of every Sybil-proof mechanism is far below the social optimum in some cases. Its proof is included in Appendix~\ref{app:sp-eff-ratio}.

Because every SP mechanism is sufficiently bad compared to social optimum, it is natural to compare their social welfare relative to other SP mechanisms.
However, this further impossibility result indicates that every SP mechanism would perform extremely worse than another SP mechanism in some cases. Therefore, we cannot find any optimal diffusion auction, even when the optimality is relative to each other. The proof can also be found in Appendix~\ref{app:sp-eff-ratio}.

\begin{theorem}\label{thm:no-best}
    For any non-deficit, SP, and IR diffusion auction mechanism \(\mechanism\), and for any \(\varepsilon > 0\), there exists another non-deficit, SP, and IR diffusion auction mechanism \(\mechanism'\) such that 
    \[\inf_{\vek\theta}\frac{\mathcal{SW}^\mechanism(\vek\theta)}{\mathcal{SW}^{\mechanism'}(\vek\theta)}<\varepsilon.\]
\end{theorem}

We can derive a similar result in terms of the seller's revenue.

\begin{theorem}\label{thm:no-best-rev}
    For any non-deficit, SP, and IR diffusion auction mechanism \(\mechanism\), and for any \(\varepsilon > 0\), there exists another non-deficit, SP, and IR diffusion auction mechanism \(\mechanism'\) such that 
    \[\inf_{\vek\theta}\frac{\mathcal{R}^\mechanism(\vek\theta)}{\mathcal{R}^{\mechanism'}(\vek\theta)}<\varepsilon.\]
\end{theorem}

The theorems above indicate that all SP mechanisms have extremely low social welfare and revenue compared to some other SP mechanisms. These impossibility results are surprising and show the drastic difference between diffusion mechanisms and traditional auctions.

\subsection{Experiments}\label{h2:experiments}
Despite the qualitative comparison results in Section~\ref{h2:comparison}, we still wonder how much our mechanisms are better than NSP, 
and how much social welfare and revenue is sacrificed for Sybil-proofness with comparison to other diffusion mechanisms.
Therefore, we conduct simulations to analyze the performance of mechanisms in the average case. 
Such experiments have never been performed on diffusion auctions in previous literature, so we have to be innovative in the settings.

To test the diffusion auction mechanisms, we must specify the private value vector of buyers and the social network structure. Each buyer's bid is a one-dimensional continuous variable and can be captured with a distribution function. For simplicity, we assume the private values are drawn i.i.d. from a uniform distribution on \([0,1]\). However, the graph structure in diffusion auctions is highly complex. Since diffusion auctions are held on social networks, we take inspirations from network science to create distributions for our input. Price's model~\cite{Price:model} is a simple and classical model for directed networks, used to describe various scale-free networks in the real world~\cite{toivonen2009comparative}. It generates a graph of \(n\) vertices, each with a degree of \(m\). Despite IDM and VCG being not Sybil-proof, we assume that all agents act truthfully in the experiment.

The mechanisms are tested with graphs with \(n=100\) vertices, and the density can be controlled by changing the parameter \(m\). For each \(m\), 1,000 inputs are generated as specified above. We test five mechanisms: NSP, STM, SCM, IDM and VCG. We calculate and analyze their social welfare and revenue. The results are visualized with box plots in Figure~\ref{fig:exp-result}.

We have the following observations. Firstly, our mechanisms achieve significantly higher social welfare and revenue than the non-diffusion NSP mechanism. Secondly, the average-case social welfare distribution of either STM or SCM is very close to the social optimum (VCG), especially when the graph is denser. Thirdly, STM has the highest revenue, which is consistent with theoretical analysis. Finally, seller's revenue of SCM is slightly lower than IDM, and higher than VCG.

Experimental results indicate that our diffusion mechanisms have significantly better performance than NSP, and we do not sacrifice seller's revenue and social welfare much to achieve Sybil-proofness.

\begin{figure}[t]
    \centering
    \begin{subfigure}[b]{0.375\linewidth}
       \includegraphics[width=1\linewidth]{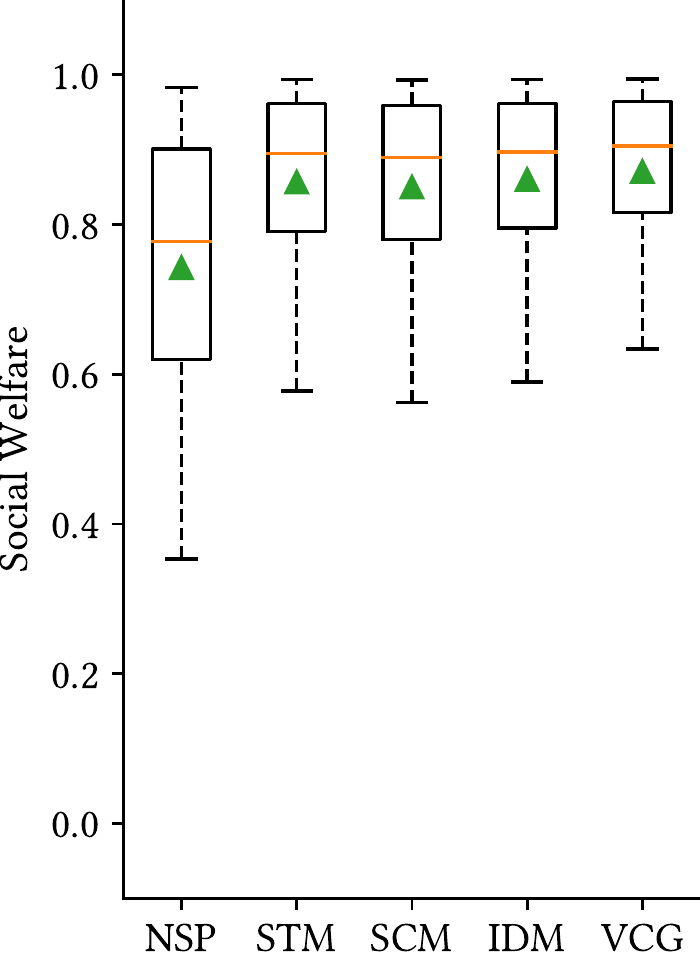}
       \captionsetup{width=.8\linewidth}
       \caption[]
        {{Social welfare of mechanisms when \(m=3\).}}
       \label{sfig:m3sw}
    \end{subfigure}
    \begin{subfigure}[b]{0.375\linewidth}
       \includegraphics[width=1\linewidth]{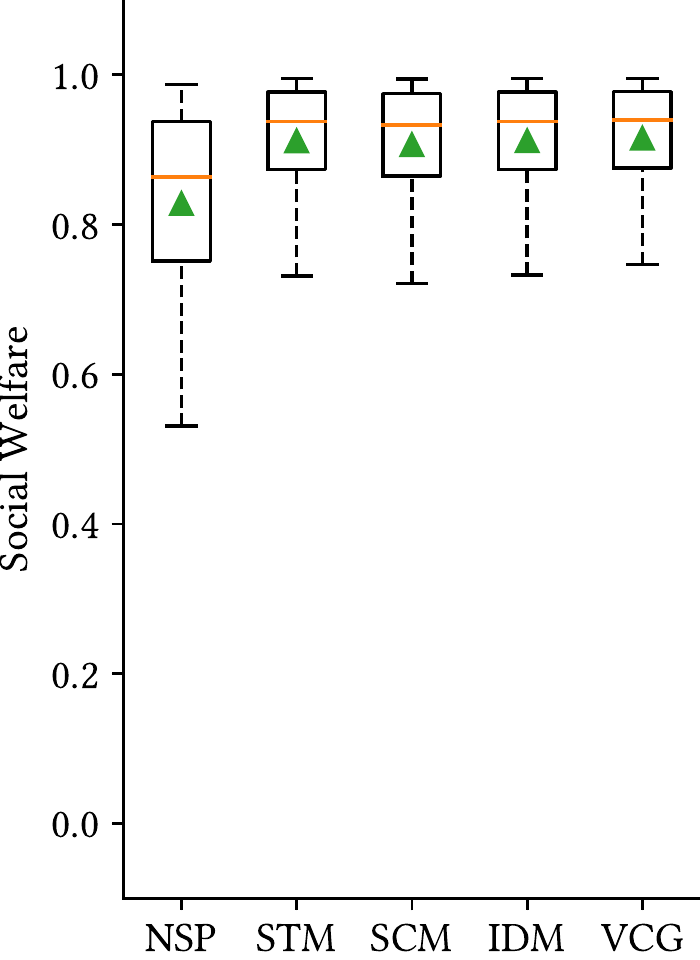}
       \captionsetup{width=.8\linewidth}
       \caption[]
        {{Social welfare of mechanisms when \(m=5\).}}
       \label{sfig:m5sw} 
    \end{subfigure}
    \begin{subfigure}[b]{0.375\linewidth}
       \includegraphics[width=1\linewidth]{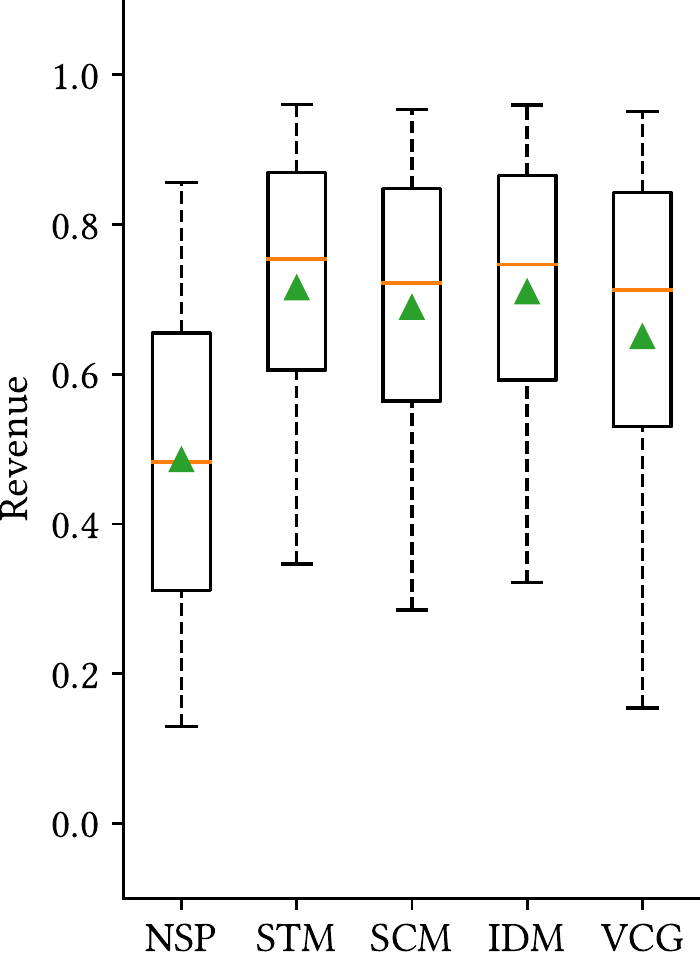}
       \captionsetup{width=.8\linewidth}
       \caption[]
        {{ Seller's revenue of mechanisms when \(m=3\).}}
       \label{sfig:m3r}
    \end{subfigure}
    \begin{subfigure}[b]{0.375\linewidth}
       \includegraphics[width=1\linewidth]{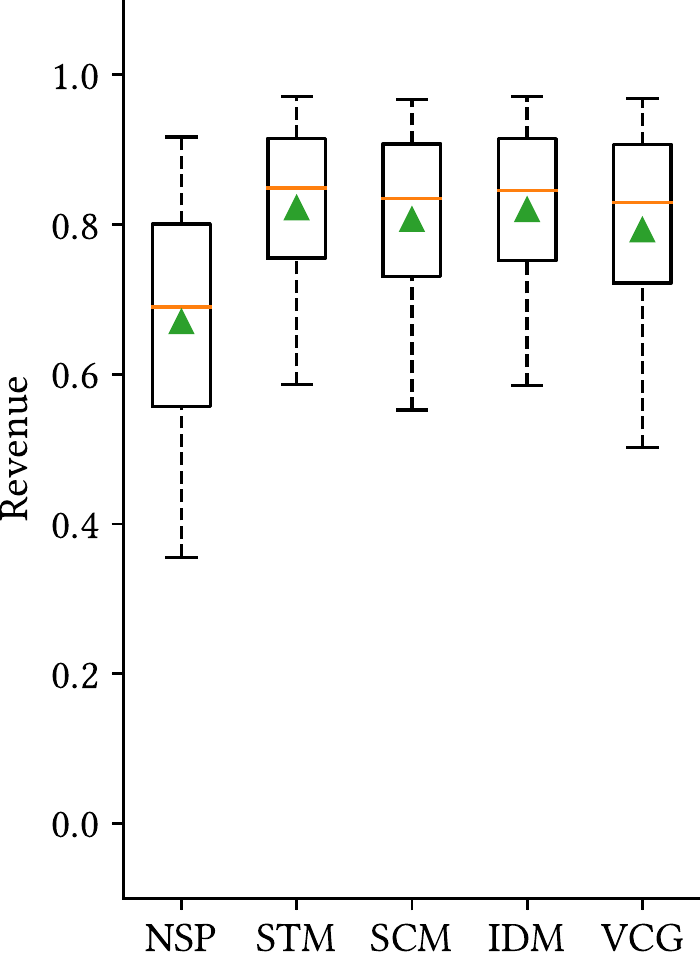}
       \captionsetup{width=.8\linewidth}
       \caption[]
        { Seller's revenue of mechanisms when \(m=5\).}
       \label{sfig:m5r} 
    \end{subfigure}
    \caption{The welfare and revenue distribution of five mechanisms on graphs of different densities. The orange line is the median, the green triangle is the mean, the box denotes the range between the first and the third quartile, and the whisker represents the range between the 5th and 95th percentile.}
    \label{fig:exp-result}
\end{figure}

\section{Conclusions}

In this paper, we study an important issue in diffusion auctions, the Sybil attack. We find that previous diffusion mechanisms are vulnerable to Sybil attacks.
We have proposed two novel solutions, STM and SCM, and proved that they are incentive compatible and Sybil-proof. We further discuss the social welfare and revenue of these two mechanisms. Theoretical analysis and experiments indicate that STM and SCM achieve Sybil-proofness with little sacrifice in the social welfare and revenue.

We also conduct worst-case analysis on all Sybil-proof diffusion mechanisms. We prove negative conclusions that the social welfare and revenue of every SP mechanism is far below some other SP mechanism in some cases.

Our work raises many open problems in the domain of Sybil-proof diffusion auctions. Firstly, how to develop Sybil-proof diffusion mechanisms for selling multiple items? Secondly, is there any other effective way to achieve Sybil-proofness? Thirdly, since we can't pick out the optimal Sybil-proof diffusion mechanism in the worst case, can we develop other methods to compare SP mechanisms? Or can we only compare a subset of all SP mechanisms to avoid such negative conclusions? Furthermore, how to reward the intermediate buyers fairly is also worth consideration.

\bibliographystyle{plainnat}
\bibliography{main}

\appendix
\section*{Appendix}
\section{Table of Notations}\label{app:notations}
The notations in this paper are organized in Table~\ref{table:notation}.

\begin{table}
\caption{Notations in this paper}
\begin{tabularx}{\linewidth}{@{}p{0.3\linewidth}X@{}}
\toprule
  \multicolumn{2}{l}{\ul{Model:}} \\
  $s$ & Seller \\
  $n$ & Number of buyers \\
  $N=\{1,2,\dots,n\}$ & Set of buyers \\
  $r(i)$ & Neighbors of \(i\), for \(i\in N\) or \(i=s\) \\
  $v_i$ & Private value of \(i\), for \(i\in N\) \\
  $\theta_i=(v_i,r(i))$ & Private type of \(i\), for \(i\in N\) \\
  $r'(i)$ & Set of neighbors that \(i\) diffuses to, for \(i\in N\) \\
  $v'_i$ & Bid of \(i\), for \(i\in N\) \\
  $\theta'_i=(v'_i,r'(i))$ & Reported type of \(i\), for \(i\in N\) \\
  $\vek\theta=(\theta_1,\dots,\theta_n)$ & Type profile \\
  $\vek\theta'=(\theta'_1,\dots,\theta'_n)$ & Report profile \\
  $\set\Theta_i=\mathbb{R}_{\geq 0}\times 2^N$ & Type space for buyer \(i\) \\
  $\set\Theta$ & Space of all possible profile\\
  \multicolumn{2}{l}{\ul{Mechanism:}} \\
  \(\pi_i(\vek\theta') \in \{0,1\}\) & Allocation function of \(i\) \\
  \(t_i(\vek\theta') \in {\mathbb{R}}\) & Payment function of \(i\) \\
  \(\mechanism=(\vek\pi(\cdot),\vek{t}(\cdot))\) & A diffusion auction mechanism \\
  \(u_i(\theta_i, \maskprofile,\mathcal{M})\) & Utility function of \(i\) under \(\mechanism\) \\
  \(\revenue^\mechanism(\vek\theta)\) & Seller's revenue under mechanism \(\mechanism\)\\
  \(\welfare^\mechanism(\vek\theta)\) & Social welfare under mechanism \(\mechanism\)\\
  \(\welfare^*(\vek\theta)\) & Optimal social welfare\\
  \multicolumn{2}{l}{\ul{Graphs and graph theoretical constructions:}} \\
  $G$ & Social network graph of \(\vek\theta\) \\
  $G(\maskprofile)$ & Social network graph of \(\vek\theta'\) \\
  $G_s(\maskprofile)$ & Subgraph of \(G(\maskprofile)\) containing only vertices reachable from \(s\) \\
  $x \dom y$ & Vertex \(x\) dominates vertex \(y\) \\
  $idom(x)$ & Immediate dominator of \(x\) \\
  $\alpha(x)$ & Set of vertices dominated by \(x\) \\
  $\Gamma_0$ & Externally provided trustworthy vertices \\
  $\Gamma(\maskprofile)$ & Guaranteed non-Sybil vertices on graph $G_s(\maskprofile)$\\
  $K_x$ & Sybil cluster of \(x\) for \(x\in \Gamma(\maskprofile)\) \\

  \multicolumn{2}{l}{\ul{Used by STM:}} \\
  \(x^*\) & Highest bidder \\
  \(\maxbid{S}\) & Highest bid in a set \(S\) \\
  \(C_x\) & Dominator sequence of vertex \(x\) \\
  \(c_j\) & The \(j\)-th element of dominator sequence $C_{x^*}$\\
  \(\beta_j\) & Set of vertices that is guaranteed not to be a Sybil identity of \(c_j\) \\
  \(p_j\) & buying price of \(c_j\) when she get the item from \(c_{j-1}\) \\
  \(q_j\) & selling price \(c_j\) gets when pass the item to \(c_{j+1}\) \\
  \multicolumn{2}{l}{\ul{Used by SCM:}} \\
  $H$ & Reconstructed social network with vertices in $\Gamma(\maskprofile)$\\
  $T_H$ & A random shortest-path tree of \(H\) \\
  $\hat{G}$ & Subgraph of \(G_s(\theta)\) that is constructed according to $T_H$\\
  \multicolumn{2}{l}{\ul{Used in the proofs:}} \\
  \(A(x)\) & The relative complement of $\alpha(x)$ with respect to vertex set $V(G_s)$\\
\bottomrule
\end{tabularx}
\label{table:notation}
\end{table}

\section{NSP, VCG and IDM}\label{app:idmvcg}
Various diffusion auction mechanisms have been proposed \cite{li:IDM} to achieve the incentive compatibility goal. In this section, we will provide three examples.

The most trivial way to achieve IC in diffusion auctions is to ignore diffusion completely, treat it as a single-item auction, and perform a second-price auction on the seller's neighbors \(r(s)\). We use \(\maxbid{S}\) to denote the highest bid in a set \(S\), or \(\maxbid{S}=\max_{x\in S} v'_x\).

\begin{framed}
 \label{mechanism:NSP}
 \noindent\textbf{Neighbor Second-Price Auction (NSP)}\\
 \rule{\textwidth}{0.5pt}

\begin{enumerate}
    \item Input the reported type profile \(\maskprofile\).
    \item The item is sold to the highest bidder among the seller's neighbor \(x^*=\arg\max_{x\in r(s)}{v'_x}\).
    \item The winner pays to the seller the second-highest bid in \(r(s)\), i.e. \(t_{x^*}=\maxbid{r(s)\setminus\{x^*\}}\).
\end{enumerate}
\end{framed}

However, this trivial mechanism defeats the purpose of diffusion auctions, i.e. to utilize the social network structure to advertise the auction to a broader audience. To fulfill this goal while still respecting incentive compatibility, we can apply the famous VCG mechanism~\cite{AGTbook2007} to our scenario. Recall that in VCG, each agent pays an amount equal to the social cost of their participation. In our scenario, an agent's diffusion may introduce new agents to the auction. Therefore, these agents' participation may increase the social welfare, instead of incurring a social cost, and should be rewarded by the mechanism.

The set \(\alpha(x)=\{y\mid x\dom y\}\subseteq V(G_s(\maskprofile))\) is defined to be the set of vertices that cannot be reached from \(s\) if \(x\) is not present. \(V(G_s(\maskprofile))\setminus \alpha(x)\) is the set of vertices that \(s\) can reach even if \(x\) is not present.

\begin{framed}
 \label{mechanism:VCG}
 \noindent\textbf{VCG Mechanism}\\
 \rule{\textwidth}{0.5pt}

\begin{enumerate}
    \item Input the reported type profile \(\maskprofile\).
    \item We can find the highest bidder among the visible agents \(x^*=\arg\min_{x\in V(G_s)}{v'_x}\). This agent wins the item. Formally\footnotemark{}, \(\pi_{x^*}=1\). She also pays the second-highest price among all visible agents \(t_{x^*}=\maxbid{V(G_s)\setminus \alpha(x)}\).
    \item The payment function of a buyer \(x\neq x^*\) is \(t_x=\maxbid{V(G_s)\setminus \alpha(x)}-\maxbid{V(G_s)}\). Note that this value is equal to the decrease in social welfare due to the participation of \(x\), which is non-positive (i.e., a reward instead of a payment).
\end{enumerate}
\end{framed}

\footnotetext{We omit the argument \((\maskprofile)\) in \(\pi_i(\maskprofile)\) and \(t_i(\maskprofile)\) when presenting mechanisms for better readability.}

We can observe that, in VCG, if an agent \(x\) receives a participation reward, she must be \emph{critical} for the diffusion to \(x^*\): Since \(x^*\in \alpha(x)\), \(s\) cannot reach \(x^*\) without the participation of \(x\). This is the exact definition of \(x \dom y\). Therefore, \(C_{x^*}\) is the exact set of agents that can receive rewards. 

In \cite{li:IDM}, Li et al. observed that VCG mechanism is IC and has optimal social welfare, but may have a budget deficit because of the payments made to each intermediate buyer. They proposed Information Diffusion Mechanism (IDM) to tackle with this problem. IDM rewards the same set of agents as VCG, but with different amounts to remain within budget.

\begin{framed}
 \label{mechanism:IDM}
 \noindent\textbf{Information Diffusion Mechanism (IDM)}\\
 \rule{\textwidth}{0.5pt}

\begin{enumerate}

    \item \label{step:idm:1} Given the reported type profile \(\maskprofile\), we can find the visible buyer with the highest bid, denoted by \(x^*\), where \(v'_{x^*}=\maxbid{V(G_s(\maskprofile))}\).
    
    \item \label{step:idm:2} Compute the dominator sequence \(C_{x^*}=\{c_0=s,c_1,\dots,c_\ell=x^*\}\). Specifically, we have \(c_j = idom(c_{j+1})\), for all \(0 \le j < \ell\).

    \item \label{step:idm:3} We define \(p_j\) as the optimal social welfare when \(c_j\) is not present, and \(q_j\) as the optimal social welfare when all agents further in the \(C_{x^*}\) sequence are not present.
\begin{align*}
    p_j &= \maxbid{V(G_s)\setminus \alpha(c_j)},\\
    q_j &= \left\{
            \begin{aligned}
            &p_{j+1}       & \text{for }j < \ell,\\
            &\maxbid{V(G_s)} & \text{for } j = \ell.
            \end{aligned}
        \right.
\end{align*}
    \item \label{step:idm:4} Pick a \(c_d\) satisfying \(v'_{c_d}\geq q_d\); when there are multiple candidates, pick the one with the lowest index \(d\). Since \(q_\ell=\maxbid{V(G_s)}=v'_{c_\ell}\), such \(c_d\) must exist. This agent wins the item with \(\pi_{c_d}=1\).
    \item \label{step:idm:5} The payment of \(c_j\) is calculated as follows.
\[
t_{c_j}(\vek\theta')=\left\{
\begin{aligned}
p_j &{} - q_j &\text{for } & 1 \le j < d,\\
p_j &{}       &\text{for } & j = d.
\end{aligned}
\right.
\]

\item \label{step:idm:6} The payment and allocation function of all other buyers are zero. The revenue of the seller is \(t_s(\vek\theta')=-\sum_{1\leq j\leq d}{t_{c_j}}\).
\end{enumerate}
\end{framed}

It is proven in \cite{li:IDM} that IDM is IC, IR, and non-deficit.

\section{Vulnerability of NRM and FPDM}\label{app:attack}
In this section, we use detailed examples to illustrate our observations that NRM~\cite{10.5555/3398761.3398947} and FPDM~\cite{zhang2021fixed} are not SP.

\begin{figure}[bht]
    \centering
    \includegraphics[width=0.66\linewidth]{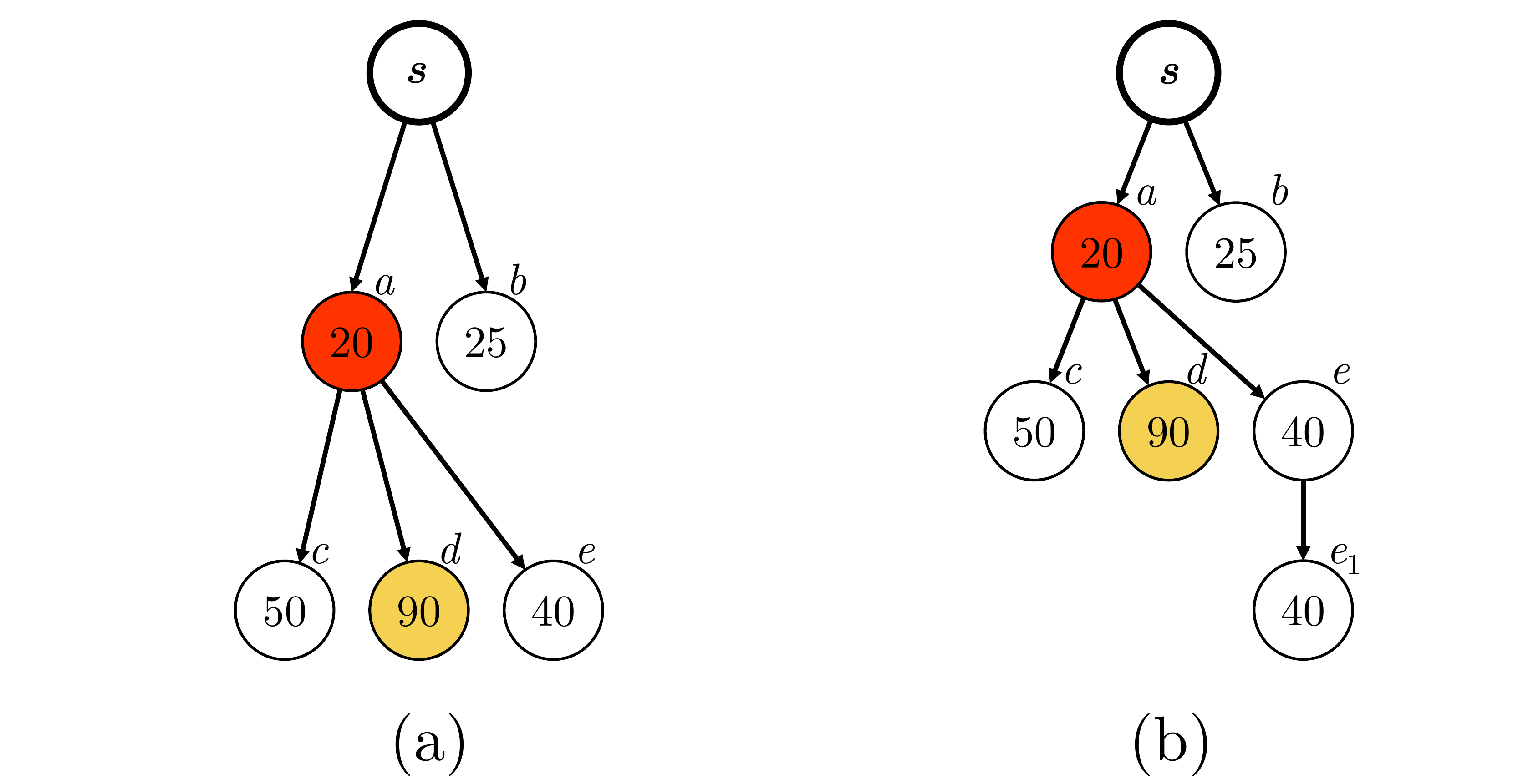}
    \caption{Sybil attack counterexamples of NRM and FPDM.}
    \label{fig:NRM}
\end{figure}

\begin{observation}
NRM is not Sybil-proof.
\end{observation}

We show this by the counterexample in Figure~\ref{fig:NRM}(a).
The seller's neighbors are $a$ and $b$, and buyer $a$ knows the existence of $c$, $d$, and $e$. Their private values are $20$, $25$, $50$, $90$, and $40$ respectively. If everyone reports truthfully, $d$ will get the item. The utilities of $a$ and $b$ are zero. Buyers $c$ and $e$ get \(\frac{1}{3} \times (40 - 25) = 5\) and \(\frac{1}{3} \times (50 - 25) = \frac{25}{3}\) units of money respectively. The winner $d$ pays $50$ for the item and receives $\frac{1}{3} \times (40 - 25) = 5$ units of money in the redistribution stage, totaling a utility of \(90 - 50 + 5 = 45\).

If $e$ creates a Sybil identity \(e_1\) to raise the percentage of nodes inside its dominant tree, as shown in Figure~\ref{fig:nonSP}(b), it will receive \(\frac{2}{4} \times (50 - 25) = \frac{25}{2}\) in the redistribution stage, leading to a higher utility for $e$. 

\begin{observation}
FPDM is not Sybil-proof.
\end{observation}

We still use the pair of examples shown in Figure \ref{fig:NRM}. Assuming that everyone reports truthfully and the fixed price for branch \(a\) is \(30\), buyers $c$, $d$, $e$ have equal opportunities to get the item. If $e$ again uses a false-name identity $e_1$, then she will get the item because she has the largest number of neighbors among the claimers.

By creating Sybil identities, the claimer can get an advantage in competing for the item.

\section{Proofs on Sybil Clustering}\label{app:clustering}

\begin{replemma}{lem:cluster-disjoint}
${\bigcup_{x \in \Gamma(\maskprofile)} K_x}$ is a (disjoint) partition of the set $V(G_s(\maskprofile))$.
\end{replemma}
\begin{proof}
First, we prove that every \(t\in V(G_s)\) belongs to some \(K_x\). Consider any simple path \([s, u_1, u_2, \dots, u_\mu, t]\) from \(s\) to \(t\); such a path must exist because every vertex is reachable from \(s\). In the sequence \(\{u_0=s,u_1,\dots,u_\mu\}\), we can find the highest index \(i\) such that \(p_i\in \Gamma\); such \(i\) must exist because \(s\in\Gamma\). Then \([u_i, u_{i+1},\dots, u_\mu, t]\) is a \(\Gamma\)-free path from \(u_i\) to \(t\), or equivalently, \(t\in K_{u_i}\).

Then, we prove that Sybil clusters are disjoint. Assume otherwise, we have \(t\in K_x\cap K_y\) where \(x \not = y\). Furthermore, we assume that among all vertices in \(K_x\cap K_y\), \(t\) has the shortest shortest-path from \(x\). Let \(U=[x, u_1, u_2, \dots, u_\mu, t]\) be a path from \(x\) to \(t\) such that none of \(u_i\) is in \(\Gamma\), and similarly, \(W=[y,w_1,w_2,\dots, w_\omega ,t]\).
There must be a \(t'=u_i=w_j\) for some \(i,j\) because \(\Gamma\) is a \emph{meeting point closure} and \(t \not \in \Gamma\).
Such \(t'\) must be in \(K_x\cap K_y\) because \([x,u_1,u_2,\dots, u_{i-1}, t']\) and \([y,w_1,w_2,\dots, w_{j-1}, t']\) are two \(\Gamma\)-free paths. However, \(t'\) has a shorter shortest-path from \(x\), contradicting the choice of \(t\). 
\end{proof}

We can find out that any inter-cluster edge must point to the root of the latter cluster.
The proofs of these lemmas are omitted here and can be found in Appendix~\ref{app:clustering}. 
\begin{lemma}\label{lem:cluster-graph-edge}
If there is a directed edge \((t,z)\in E(G_s(\maskprofile))\) between \(t\in K_x\) and \(z\in K_y\), where \(x\neq y\), then \(z=y\).
\end{lemma}

\begin{proof}
We assume otherwise that \(z\neq y\). We can find two paths \(U=[x, u_1, u_2, \dots, u_\mu, t]\) and \(W=[y,w_1,w_2,\dots, w_\omega ,z]\) because \(t\in K_x\) and \(z\in K_y\). The path \(U \| z=[x, u_1, u_2, \dots, u_\mu, t,z]\) and the path \(W\) form two vertex-disjoint paths from \(x\) and \(y\) resp. to \(z\). By the definition of \(\Gamma\), we have \(z\in \Gamma\), so \(z\) must be the root of its component, which contradicts the assumption.
\end{proof}

\section{Proofs on STM}\label{app:stm}
\subsection{The Proof of Theorem~\ref{thm:main}}
To ease our expression, we categorize all buyers into four roles by their outcome: the winner \(c_d\), the brokers \(C_{c_d}\setminus\{s, c_d\}\), the shadowed \(\alpha(c_d) \cap C_{x^*} \setminus \{c_d\}\), and the other buyers considered unrelated. We assume that every buyer mentioned below is reachable from \(s\); otherwise, their reports will not affect the outcome.

Some notations are introduced in the proof.

On the dominator sequence \(C_{x^*}=\{c_0=s,c_1,\dots,c_\ell=x^*\}\), a buyer \(c_j\) is said to be a candidate winner if \(v'_{c_j}\geq q_j\). The final winner is the candidate winner with the lowest index.

It is defined previously that $$\alpha(x)=\{y\mid x\dom y\}.$$

We denote the complementary set of \(\alpha(x)\) as $$A(x)=V(G_s)\setminus \alpha(x).$$

Intuitively, set \(\alpha(x)\) represents the possible identities of \(x\) without prior information, and \(\gamma(x)\) is the subset of \(\alpha(x)\) that must not be Sybil identities of \(x\) with \(\Gamma\) given. 

\begin{lemma}\label{lem:*spdm-ir}
{\STM} is individually rational.
\end{lemma}
\begin{proof}
Because \(V(G_s)\setminus \alpha(c_j)\subseteq (V(G_s)\setminus \alpha(c_j))\cup \setbeta{j}\subseteq V(G_s)\setminus \alpha(c_{j+1})\), we can see that \(p_j\leq q_j\leq p_{j+1}\).

When a buyer \(i\) bids truthfully, her utility is \(v'_i\pi_i - t_i=v_i\pi_i - t_i\).
If she is shadowed or unrelated, her allocation and payment are defined to be zero, thus non-negative.
If \(i\) is a broker, i.e.,  \(i=c_j\), then \(t_i=p_j-q_j \leq 0\) and \(\pi_i=0\), summing up to be a non-negative utility.
If \(i=c_d\) is the winner, \(v'_{c_d}\geq q_d\geq p_d\) is satisfied, guaranteeing \(v'_{c_d}\geq p_d=\maxbid{A(c_d)}\); therefore, the utility \(v_{c_d} - p_d = v'_{c_d} - p_d \geq 0\).
\end{proof}

\begin{lemma}\label{lem:*spdm-ic}
{\STM} is incentive compatible. 

Formally, for any type profile \(\vek\theta\), any buyer \(i\in N\) and any \(\theta'_i \in \settype{i}\) satisfying \(r'(i)\subseteq r(i)\), we have \(v_i\cdot \pi_i(\vek\theta) - t_i(\vek\theta)\geq v_i\cdot \pi_i(\maskprofile) - t_i(\maskprofile)\), where \(\maskprofile=(\theta'_i,\action{-i})\).
\end{lemma}
\begin{proof}
Let us split cases by the role of \(i\) when she reports truthfully.

A shadowed buyer \(c_t\) with \(d < t \le \ell\) cannot escape zero utility. Her non-diffusion may change the highest bid the mechanism can see, but as long as \(\tilde{x}^*\in \alpha(c_t)\), she will remain shadowed. Otherwise, if \(\tilde{x}^* \not\in \alpha(c_t)\), \(c_t\) will become unrelated. 

For an unrelated buyer \(i\), her diffusion choice does not affect the selection of \(x^*\) because \(i\) does not dominate \(x^*\); therefore, the only way she can become related is to bid higher than \(x^*\), which results in an unaffordable price of \(\maxbid {A(i)}\geq v'_{x^*}\geq v'_i=v_i\).

For a broker \(i=c_j\), her utility is \(\maxbid{A(c_j)\cup \setbeta{j}} - \maxbid{A(c_j)}\). If she hides some of her neighbors in \(r(c_j)\) but still be a broker, the difference between \(p_j\) and \(q_j\) will not increase because \(\setbeta[\tilde]{j}\) will become smaller but \(A(c_j)\) remains unchanged. Therefore, if \(i\) remains a broker, reporting truthfully is not worse. Changing role 
will not be good: a shadowed or unrelated buyer implies a zero utility, and being a winner lessens \(c_j\)'s utility because \(v_{c_j} < \maxbid{A(c_j)\cup \setbeta{j}}\).

For a winner \(i=c_d\), the utility is independent of her report as long as she retains her winner role. Becoming shadowed or unrelated results in a zero utility. Becoming a broker by lowering her bid to \(q_d\) will lower her utility from \(v_i-p_d\) to \(\tilde q_d-p_d\),  where \(\tilde q_d < v_i\).
\end{proof}

Lemma \ref{lem:*spdm-ic} implies incentive compatibility.
\begin{lemma}\label{lem:*spdm-sp}
Under {\STM}, using Sybil identities does not increase utility.

Formally, for all \(i\in N\), all \(\action{-i}\), all \(\theta_i'\)
and all \(\theta'_{i_1},\dots,\theta'_{i_k}\), there exists a report \(\tilde{\theta}'_i\) such that \(v_i\cdot\pi_i(\tilde{\theta}'_i ,\action{-i}) - t_i(\tilde{\theta}'_i ,\action{-i})\geq v_i\cdot \pi_i^+ - t_i^+\), where \(\maskprofile=(\action{-i},\theta'_{i},\theta'_{i_1},\dots,\theta'_{i_k})\), \(\pi_i^+=\pi_i(\maskprofile)+\sum_{1\leq j\leq k}\pi_{i_j}(\maskprofile)\), and \(t_i^+=t_i(\maskprofile)+\sum_{1\leq j\leq k}t_{i_j}(\maskprofile)\).
\end{lemma}
\begin{proof}
Let \(\phi=\{i, i_1, \dots, i_k\}\) be the set of both real identities that \(i\) uses. We can take the report \(\tilde{\theta}'_i =(\tilde{v}'_i,\tilde{r}(i))\), where \(\tilde{v}'_i = \max_{x\in \phi} v'_x \), and \(\tilde{r}(i) = \bigcup_{x\in\phi}r'(x)\setminus\phi\).

Set \(\Gamma\) remains unchanged because the externally provided \(\Gamma_0\) and \(r(s)\) are constant.

Assume \(i\) has a non-zero utility (otherwise, it's covered by Lemma~\ref{lem:*spdm-ir}), then at least one identity \(x \in\phi\) is on the dominator sequence \(C_{x^*}\). Using graph-theoretic properties of Sybil identities, it is easy to find that \(\phi\cap C_{x^*}\) occupies a consecutive sub-sequence of the dominator sequence \(C_{x^*}\), namely \([c_a, c_{a+1},\dots,c_b]\) with \(c_a=i\).

The proof is given by cases on the roles of the identities of \(i\).

Case 1: if \([c_a,\dots,c_b]\) contains the winner, or equivalently \(a \leq d \leq b\).
When {\STM} is performed on the new reported profile \((\tilde{\theta}_i,\action{-i})\), \(\tilde{x}^*\) either equals to \(i\) (if \(x^*\in\phi\) before) or remains unchanged (and \(i\) remains on the dominator sequence as \(\tilde{x}_a\)). Either way, \(i\) will be the new winner since \(\tilde{v}'_i = \max_{x\in\phi} v'_{x} \geq v'_{c_d}\geq \maxbid{A(c_d)\cup \setbeta{d}}\geq \maxbid{\tilde{A}(\tilde{c}_a)\cup \setbeta[\tilde]{a}}\). This means \(i\) still wins the item. Considering
\begin{align*}
    t_i^+
    ={}& \sum_{a\leq j\leq d}{t_{c_j}}\\
    ={}& \sum_{a\leq j < d}{\left(p_j-q_j\right)}+p_d\\
    ={}& p_a+\sum_{a\leq j < d}{\left(p_{j+1}-q_j\right)}\\
    \geq{}& p_a=\maxbid{A(i)}=\maxbid{\tilde{A}(i)}=\tilde{t}_i,
\end{align*}
the payment \(i\) receives has a non-negative gain when she does not perform a Sybil attack.

Case 2: if none of \(c_a, c_{a+1},\dots,c_b\) is the winner, they must be all brokers.
Under the new report profile {\((\tilde{\theta}_i,\action{-i})\)}, \(x^*\) remains unchanged, and \(i\) will be a broker \(\tilde{c}_a\). She will become the immediate dominator of \(c_{b+1}\) because \(c_{a+1}, \dots, c_b\) are removed, and therefore \(\tilde{c}_{a+1}=c_{b+1}\). Because \(\tilde\alpha(\tilde{c}_{a+1})=\alpha(c_{b+1})\) remains unchanged, we can see \(\tilde{A}(\tilde{c}_{a+1})=A(c_{b+1})\setminus (\phi \setminus \{i\})\). In addition to \(\tilde{v}'_{i} = \maxbid{\phi}\), we have \(\maxbid{\tilde{A}(\tilde{c}_{a+1})}=\maxbid{A(c_{b+1})}\).

For \(j = a, \dots, b\) we can define \(B_j = \setbeta{a}\cup \setbeta{a+1}\cup\dots\cup \setbeta{j}\). By induction, we can prove that \(B_{j-1} \subseteq B_j\subseteq A(c_j)\cup \beta_{j}\). From the inequality on real numbers \(\max\{x,y,z\}-\max\{y,z\}\leq \max\{x,y\}-\max\{y\}\) we can deduce that 
\begin{align*}
&\maxbid{A(c_j)\cup \setbeta{j}}-\maxbid{A(c_j)}\\
\leq{}& \maxbid{A(c_a)\cup B_{j-1}\cup\setbeta{j}}-\maxbid{A(c_a)\cup B_{j-1}}\\
={}&\maxbid{A(c_a)\cup B_{j}}-\maxbid{A(c_a)\cup B_{j-1}}.\\
\end{align*}
Therefore,
\begingroup
\allowdisplaybreaks
\begin{align*}
    t_i^+
    ={}&    \sum_{a\leq j\leq b}{t_{c_j}}\\
    ={}&    \sum_{a\leq j\leq b}{\left(\maxbid{A(c_j)}-\maxbid{A(c_j)\cup \setbeta{j}}\right)}\\
    \geq{}&    \maxbid{A(c_a)}-\maxbid{A(c_a)\cup B_a}\\
       &    +\sum_{a < j\leq b}{\left(\maxbid{A(c_a)\cup B_{j-1}}-\maxbid{A(c_a)\cup B_j}\right)}\\
    ={}&    \maxbid{A(c_a)}-\maxbid{A(c_a)\cup B_b}.
\end{align*}
\endgroup

We are certain that \(B_b\subseteq A(c_a)\cup \gamma(c_a)\) because every \(\setbeta{j}\) is a subset of \(\gamma(c_a)\). Together with the fact that \(B_b\subseteq A(c_{b+1})\), we have \(B_b\subseteq (A(c_a)\cup \gamma(c_a))\cap A(c_{b+1})\). Therefore,
\begin{align*}
    t_i^+
    \geq{}&    \maxbid{A(c_a)}-\maxbid{A(c_a)\cup B_b}\\
    \geq{}& \maxbid{A(c_a)}-\maxbid{(A(c_a)\cup \gamma(i))\cap  A(c_{b+1})}\\
    ={}&    \maxbid{\tilde{A}(c_a)}-\maxbid{(\tilde{A}(\tilde{c}_a)\cup \tilde{\gamma}(i))\cap \tilde{A}(\tilde{c}_{a+1})}\\
    ={}&    \maxbid{\tilde{A}(c_a)}-\maxbid{\tilde{A}(\tilde{c}_a)\cup (\tilde{\gamma}(i)\cap \tilde{A}(\tilde{c}_{a+1}))}\\
    ={}&    \tilde{t}_i.
\end{align*}

In both cases, the utility achieved by not using Sybil identities is no less than using them.
\end{proof}

With Lemma \ref{lem:*spdm-ir}, \ref{lem:*spdm-ic}, and \ref{lem:*spdm-sp} above and the fact that the winner's payment and the tax between brokers are non-negative, we can conclude that

\begin{reptheorem}{thm:main}
{\STM} is non-deficit, individually rational and Sybil-proof.\qed
\end{reptheorem}

\subsection{The Proof of Lemma~\ref{lem:0spdm-zero}}
\begin{replemma}{lem:0spdm-zero}
In \emph{{\STM}} when \(\Gamma_0=\emptyset\), every buyer, except the item winner, has a zero payment, and thus, a zero utility. In an essence, every possible profit of the brokers is taxed by the seller.
\end{replemma}

\begin{proof}

When \(\Gamma_0\) is an empty set, \(\Gamma\) will only contain the children of \(s\) on the dominator tree, i.e. \(\Gamma=\{i\}\cup\{x\mid idom(x)=i\}\), so every \(\gamma(x)\) and every \(\setbeta{j}\) will be empty. For any broker \(c_j\), the set \(A(c_j)\cup \setbeta{j}\) would be equal to \(A(c_j)\), thus \(q_j=p_j\) implying that \(t_{c_j}=0\). Other buyers' utilities are defined as zero.
\end{proof}

\section{Proofs on SCM}\label{app:scm}

To investigate the properties of a randomized mechanism, we here re-express it as a probability distribution of deterministic mechanisms.

Let \(M\) be the set of all possible identities, where \[M=\{s\}\cup\left(\bigcup_{i\in N}\{i, i_1,\allowbreak i_2,\dots\}\right).\] For each pairs of identities \((x,y)\), let \(f(x,y)\) be a random variable with a uniform distribution on \((0,1)\), where all the random variables \(\{f(x,y)\mid x,y\in M\}\) are independent. When we are choosing a parent in \(T_H\) for a vertex \(x\) in \(\Gamma(\maskprofile)\), we choose the vertex \(y\) such that \((y,x)\in E(T)\) and  \(dis_y(H)+1=dis_x(H)\); if multiple such candidates exist, we choose the one with the lowest \(f(x,y)\). The shortest-path tree \(T_H\) formed by this process has the same distribution as in section \ref{h1:scm}.

Because the distribution of \(f\) is independent of the input of the mechanism, we can perceive {\SCM} as a family of deterministic mechanisms, with one \(\mathrm{SCM}_f\) for every possible \(f : M\times M\to (0,1)\). Throughout our analysis, we will prove that some inequalities hold for every \(\mathrm{SCM}_f\), which implies that incentive compatibility and Sybil-proofness hold uniformly. Furthermore, we will only consider cases where all the values of \(f(x,y)\) (\(x \in M, y \in M\)) are distinct because the event that some two \(f\) values coincide will happen with a probability of \(0\).

\begin{lemma}\label{lem:*aspdm-sp}
   Under all derandomized cases of {\SCM}, using Sybil identities does not increase utility.

Formally, for all \(f : M\times M\to (0,1)\) without ties, under mechanism \(\mathrm{SCM}_f=(\pi,t)\), for all \(i\in N\), all \(\action{-i}\),
and all \(\theta'_{i},\theta'_{i_1},\dots,\theta'_{i_k}\), there exists a report \(\tilde{\theta}'_i\)
such that \(v_i\cdot \pi_i(\tilde{\theta}'_i ,\action{-i}) - t_i(\tilde{\theta}'_i ,\action{-i})\geq v_i\cdot \pi_i^+ - t_i^+\), where \(\maskprofile=(\action{-i},\theta'_{i}, \theta'_{i_1},\dots,\theta'_{i_k})\), \(\pi_i^+=\pi_i(\maskprofile)+\sum_{1\leq j\leq k}\pi_{i_j}(\maskprofile)\), and \(t_i^+=t_i(\maskprofile)+\sum_{1\leq j\leq k}t_{i_j}(\maskprofile)\).
\end{lemma}

\begin{proof}
    We construct the same non-Sybil-attacking report as Lemma~\ref{lem:*spdm-sp}, i.e., we take \(\tilde{\theta}'_i=(\tilde{v}'_i,\tilde{r}(i))\), where \(\tilde{v}'_i = \max_{x\in \phi} v'_x \), \(\tilde{r}(i) = \bigcup_{x\in\phi}r'(x)\setminus\phi\), and \(\phi=\{i, i_1, \dots, i_k\}\).
    
    Similar to Lemma~\ref{lem:*spdm-sp}, \(\Gamma\) does not change under the new report profile.
    
    Notice that given a specific \(f\) and report profile \(\action{-i}\), some inter-cluster edges from \(\phi\) to some \(y \in \Gamma\) are impossible to remain on \(\hat{G}\) due to a lower \(f\) value of another edge. We say that \(y\) is \emph{shielded} from \(i\) and denote the set of all shielded vertices \(S_i\).
    
    On the subgraph \(\hat{G}\) where STM is performed, the remaining edges from \(i\)'s identities are \(\{(x, y) \mid x \in \phi, y \in r'(x) \setminus S_i\}\). This is equivalent to a fictional case where the neighbor set of \(i\) is restricted to \(r(i)\setminus S_i\). In this fictional case, every \(x \in \phi\) reports \(r''(x) = r'(x) \setminus S_i\).
    
    Similarly, on subgraph \(\tilde{G}\), we have \(i\) report \(\tilde{r}''(i) = \tilde{r}'(i) \setminus S_i\) in the fictional case. We then can apply Lemma~\ref{lem:*spdm-sp} to the fictional case.
\end{proof}

\begin{lemma}\label{lem:*aspdm-ic}
   Under the {\SCMfull}, if a buyer is not allowed to perform Sybil attacks, then bidding and diffusing truthfully is a dominant strategy. 
   
   Formally, for all possible \(f : M\times M\to (0,1)\) without ties, under the mechanism \(\mathcal{M}_f\), for any \(i\in N\) with arbitrary \(\theta_i = (v_i, r(i))\), any \(\action{-i}\) and any \(\theta_i'\in\settype{i}'\), we have \(v_i\pi_i(\theta_i,\action{-i})-t_i(\theta_i,\action{-i})\geq v_i\pi_i(\theta_i',\action{-i})-t_i(\theta_i',\action{-i})\). 
\end{lemma}
\begin{proof}
    Similar to the proof of Lemma~\ref{lem:*spdm-ic}, we split cases by the role of \(i\) when she reports truthfully.
    
    A shadowed buyer \(c_t\) with \(t > d\) cannot escape zero utility, regardless of what action she takes. 
    
    An unrelated buyer cannot become related unless she bids higher, which causes a negative utility.
    
    For a winner \(i = c_d\), her utility is \(v_i - \maxbid{A(i)}\). If \(i\) changes her report but still be a winner, \(v_i\) does not change, but \(\maxbid{A(i)}\) may only increase. This is because \(i\)'s under-diffusion could make \(\tilde{A}(i)\) a superset of \(A(i)\). This can only decrease the utility of \(i\). If \(i\) turns into a broker, we have \(v_i - \maxbid{A(i)} \ge q_d - p_d \ge \tilde{q_d} - \tilde{p_d}\). This is because, by misreporting, \(\beta_d\) might shrink while \({A}(x_d)\) might extend, resulting in a lower \(\tilde{q}_d - \tilde{p}_d\).
    
    For a broker \(i = c_j\), her utility is \(q_j - p_j = \maxbid{A(c_j) \cup \beta_j} - \maxbid{A(c_j)}\). Similar to the analysis above, if she remains to be a broker, its utility will not increase since \(q_j - p_j \ge \tilde{q_j} - \tilde{p_j}\). If she turns into a winner, we have \(v_i - \maxbid{\tilde{A}(i)} \le v_i - \maxbid{A(i)} \le q_j - p_j\).
    
    Therefore, a buyer cannot increase her utility by under-diffusion or misreporting her bid.
\end{proof}

With Lemma \ref{lem:*aspdm-sp} and Lemma \ref{lem:*aspdm-ic}, we can conclude that 
\begin{reptheorem}{thm:scm-main}
{\SCM} is non-deficit, IR and Sybil-proof. \qed
\end{reptheorem}

\section{The Proofs of Comparison among Mechanisms}\label{app:compare}

\begin{reptheorem}{thm:compare_STM&SCM}
When \(\Gamma_0=\emptyset\), it is guaranteed that
\begin{gather*}
\revenue^\mathrm{STM}(\vek\theta) \ge \revenue^\mathrm{SCM}(\vek\theta) \ge \revenue^\mathrm{NSP}(\vek\theta),\\
\welfare^\mathrm{STM}(\vek\theta) \ge \welfare^\mathrm{SCM}(\vek\theta) \ge \welfare^\mathrm{NSP}(\vek\theta)
\end{gather*}
for all \(\vek\theta\).
\end{reptheorem}
\begin{proof}
Recall that $G_s(\maskprofile)$ is the graph of reachable vertices under type profile \(\vek\theta\), and \(\hat{G}\) be the graph constructed by the edge-removing process of SCM. 

We first compare STM and SCM.

    The highest bidder \(x^*\) does not change between SCM and STM. Because \(\hat{G}\) is a spanning subgraph of \(G_s\), the dominator sequence of \(x^*\) on \(G_s\) is a subsequence of its dominator subsequence on \(\hat{G}\), and we have \(\alpha_{G_s}(x)\subseteq \alpha_{\hat{G}}(x)\) for every buyer \(x\). 
    
    Any candidate winner \(x\) under STM satisfies \(v_x\geq \maxbid{V(G_s)\setminus \alpha_{G_s}(x)}\), so she is still a candidate winner under SCM. The candidate winner with the lowest index has the lowest bid because set \(\alpha(c_j)\) is monotonically decreasing by \(j\).
    Therefore, the winner of SCM has the lowest bid among all candidate winners of SCM, hence not higher than the winner of STM.
    
    Under both STM and SCM, the winner \(x\) pays \(\maxbid{V(G_s)\setminus \alpha(x)}\). We say the winner \(x_1\) of SCM pays a lower price than the winner \(x_2\) of STM because \(x_1 \in \alpha_{\hat{G}}(x_2)\) and \(\alpha_{G_s}(x_1) \subseteq \alpha_{\hat{G}}(x_1) \subseteq \alpha_{\hat{G}}(x_2)\). With an empty \(\Gamma_0\), STM does not pay other buyers any money. Therefore, the seller under STM has a higher earning and a zero expense to incentivize diffusion, thus a higher revenue than SCM.
    
Now, we compare SCM with NSP.

    The neighbor set of $s$ in $G_s$, denoted as $r_{G_s}(s)$, is equal to the neighbor set of $s'$ in $\hat{G}$, denoted as $r_{\hat{G}}(s)$. This is because the arc-removing process follows the random shortest-path tree $T'$.
    
    The winner under SCM, denoted as $x$, satisfies \(v'_x\geq \maxbid{V(\hat{G})\setminus \alpha_{G_s}(x)}\). When $x \not\in r_{\hat{G}}(s)$, since a vertex in $r_{\hat{G}}(s)$ does not exist in the $\alpha_{G_s}(\cdot)$ set of anyone other than seller $s$ and herself, we claim that $v'_x \ge \maxbid{r_{\hat{G}}(s)} = \maxbid{r_{G_s}(s)}$.
    
    The revenue under NSP is the second-highest bid of these bidders, while the revenue of SCM is no less than $p_1 = \maxbid{V(\hat{G})\setminus \alpha_{\hat{G}}(c_1)}$. When $c_1 \not \in r_{\hat{G}}(s)$, $r_{\hat{G}}(s) \subseteq V(\hat{G})\setminus \alpha_{\hat{G}}(c_1)$, so
    \begin{align*}c_1 \ge \text{the second price of }r_{\hat{G}}(s) = \text{the second price of }r_{G_s}(s).\end{align*} 
    When $c_1 \not \in r_{\hat{G}}(s) = r_{G_s}(s)$, whether she is the highest bidder among the neighbors or not, the highest bid among other neighbors is no less than the second price. Therefore, SCM achieves a higher revenue than NSP.
\end{proof}

\begin{reptheorem}{thm:revenue_compare_STM}
When \(\Gamma_0=\emptyset\), for all \(\vek\theta\), we have
\begin{gather*}
\revenue^\mathrm{STM}(\vek\theta) \ge \revenue^\mathrm{IDM}(\vek\theta) \ge \revenue^\mathrm{VCG}(\vek\theta)\\
\welfare^\mathrm{SCM}(\vek\theta) \le \welfare^\mathrm{STM}(\vek\theta) \le \welfare^\mathrm{IDM}(\vek\theta) \le \welfare^\mathrm{VCG}(\vek\theta).
\end{gather*}
\end{reptheorem}

\begin{proof}
     Under these three mechanisms, the reachable graph $H$ does not change, so do the highest bidder $x^*$ and diffusion critical sequence $C_{x^*} = \{c_0 = s, c_1, \dots, c_l = x^*\}$.
  
     In IDM, the revenue of the seller stays at $p_1$ no matter who wins the item.
     
     In STM, the revenue is $$\revenue(\vek\theta')=\sum_{1\leq j\leq d}{t_{c_j}} = p_d + \sum_{1 \leq j < d}{(p_j - q_j)} = p_1 + \sum_{1 \leq j < d}({p_{j+1} - q_j)} \ge p_1.$$
     
     In VCG, the payment of non-critical buyers is zero. For critical buyer $c_j$ where \(1 \le j < d\), the payment of $c_j$ is $v'_{x^*} - p_j > 0$, which means that the mechanism should reward $c_j$ for her contribution of introducing $x^*$. The revenue is $$p_d - \sum_{1 \le j < d}{\left(v'_{x^*} - p_j\right)} = \sum_{1 \le j \le d}{p_j} - (d-1)v'_{x^*} = p_1 + \sum_{2 \le j \le d}{\left(p_j - v'_{x^*}\right)} \le p_1.$$
     
     VCG mechanism achieves the optimal social welfare, i.e., $\welfare^\mathrm{VCG}(\vek\theta)=v'_{x^*}$. Notice that any candidate winner $x$ under IDM satisfies $v_x \ge \maxbid{V(G_s)\setminus \alpha_{G_s}(x)}$, so she is still a candidate winner under STM. Among all the candidate winners under STM, the one with the lowest index in diffusion critical sequence (namely, the winner under STM) has the lowest bidder. Therefore, the winner of STM has a lower bid than the winner of IDM. 
\end{proof}

The following theorem shows that there is no clear-cut comparison of revenue between SCM and IDM, or between SCM and VCG, by giving concrete examples.

\begin{reptheorem}{thm:no_clear}
When \(\Gamma_0=\emptyset\), $\exists$ report profile $\vek\theta_1, \vek\theta_2,$ such that
\begin{gather*}
\revenue^\mathrm{SCM}(\vek\theta_1) > \revenue^\mathrm{IDM}(\vek\theta_1), 
\revenue^\mathrm{SCM}(\vek\theta_1) > \revenue^\mathrm{VCG}(\vek\theta_1), \\ \revenue^\mathrm{SCM}(\vek\theta_2) < \revenue^\mathrm{IDM}(\vek\theta_2), \revenue^\mathrm{SCM}(\vek\theta_2) < \revenue^\mathrm{VCG}(\vek\theta_2).
\end{gather*}
\end{reptheorem}
\begin{proof}
    \begin{figure}
    \centering
    \includegraphics[width=0.85\linewidth]{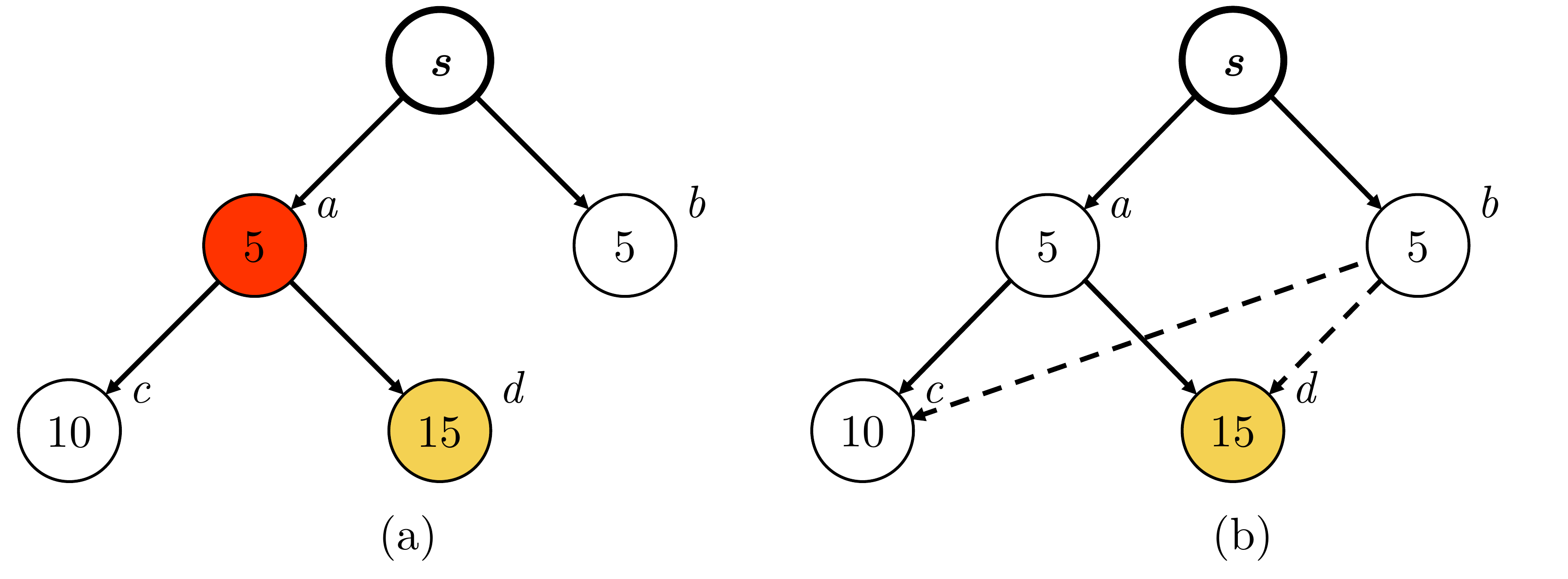}
    \caption{The examples used to prove Theorem \ref{thm:revenue_compare_STM}. The dotted lines are the arcs removed by SCM.}
    \label{fig:no-clear-cut}
    \end{figure}
    
    We set \(\vek\theta_1\) to be the case in Figure \ref{fig:no-clear-cut}(a) and \(\vek\theta_2\) as in Figure \ref{fig:no-clear-cut}(b).
    
    Under \(\vek\theta_1\), the seller $s$ only knows the existence of $a$ and $b$, agent $a$ has neighbors $c$ and $d$. The private value of $a, b, c, d$ are $5, 5, 10, 15$ respectively. Under SCM, the winner $d$ pays $10$ to get the item and other agents get zero utility. The revenue of the seller is $10$. In IDM, the agent $a$ gets a reward of $10-5=5$ for introducing $d$, and the seller's revenue is $5$. In VCG, $d$ gets the item and pays the second-highest price $10$, $b$ and $c$ gets zero utility while $a$ gets a reward of $15 - 5 = 10$ for her contribution to improving social welfare. The revenue of the seller under VCG is zero. This case shows that there exist some report profiles such that the SCM mechanism achieves a strictly higher revenue than IDM and VCG. 
    
    Under \(\vek\theta_2\), as shown in Figure \ref{fig:no-clear-cut}(b), $s$ knows $a$ and $b$, and both of them have a connection to $c$ and $d$. The private value of each buyer is the same as \(\vek\theta_1\). We assume that SCM removes the arcs $(b, c)$ and $(b, d)$ and runs STM on the remaining subgraph. Then agent $a$ will get a reward of $10 - 5 = 5$ because $c$ is in $\Gamma$. Winner $d$ pays $10$ for the item, so the revenue under SCM is $10 - 5 = 5$. In IDM, the seller gets the second price $10$ because the immediate dominator of $d$ is $s$. The revenue of the seller under VCG is also $10$. In this case, SCM achieves a strictly lower revenue than IDM and VCG. 
\end{proof}

\section{The Proof of the Impossibility Theorems}\label{app:sp-eff-ratio}
\begin{reptheorem}{thm:sp-eff-ratio}
    The worst-case efficiency ratio of any non-deficit, IR, and Sybil-proof diffusion auction mechanism is zero.
\end{reptheorem}

\begin{proof}
    Assume otherwise, a Sybil-proof diffusion mechanism \(\mechanism\) can achieve a non-zero efficiency ratio lower bound of \(\eta > 0\). We can observe that, if a buyer's private value is sufficiently high, i.e. \(\frac{1}{\eta}\) times higher than the private value of any other buyer, this buyer must win the item.
    
\begin{figure}
    \centering
    \includegraphics[width=0.85\linewidth]{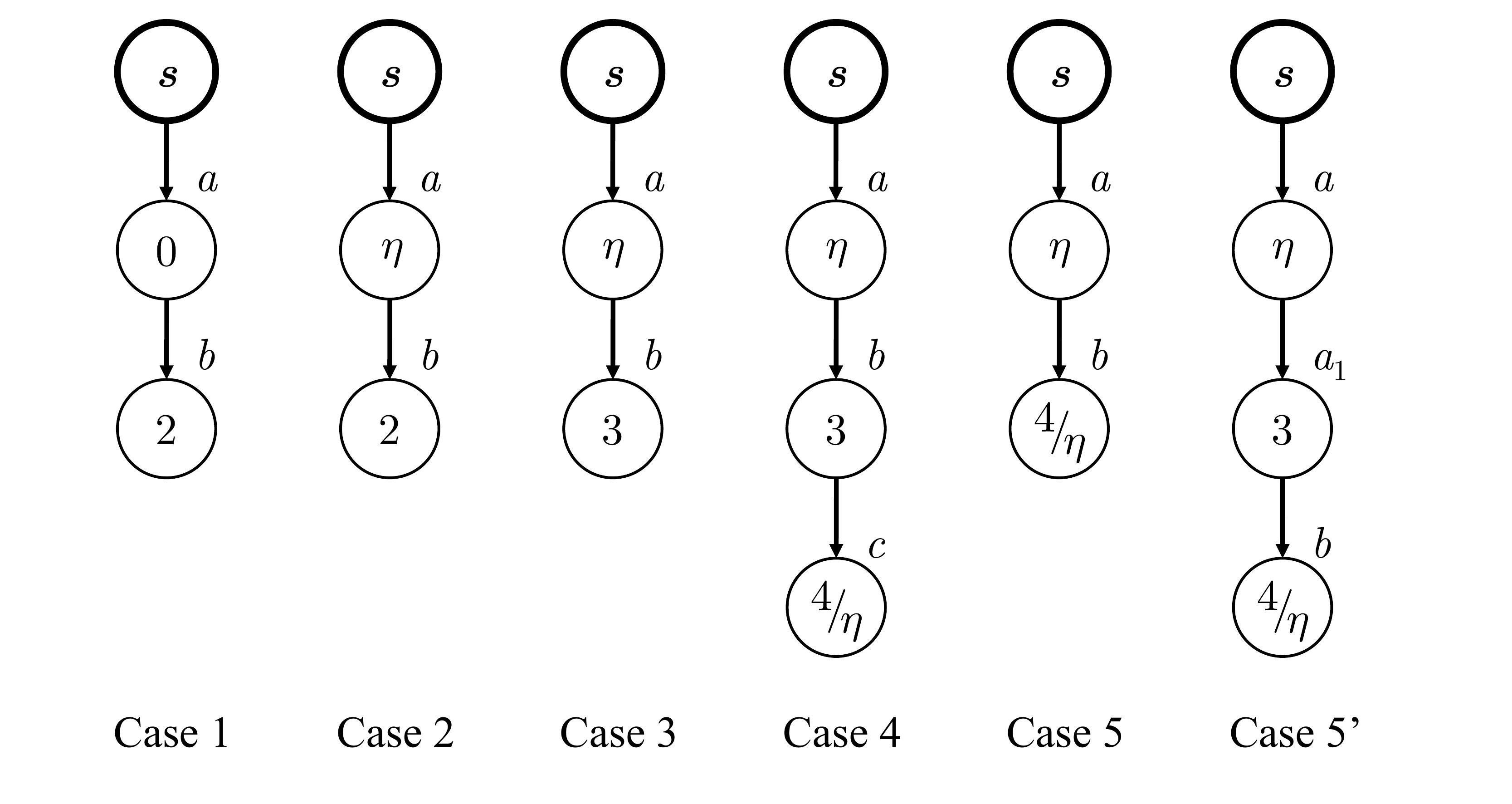}
    \caption{The social network structure used to prove Theorem \ref{thm:sp-eff-ratio}.}
    \label{fig:eff-ratio}
\end{figure}

    We consider the social network graphs in Figure \ref{fig:eff-ratio}. In the following proof, we always have \(V=\{s,a,b,c\}, r(s)=\{a\},r(a)=\{b\},r(c)=\emptyset\).

    \textbf{Case 1.} When \(r(b)=\emptyset\), \(v_a=0\), and \(v_b=2\), the item must be sold to the agent \(b\), and the prize for her must be zero (otherwise, she can bid below that prize and still win the item). By non-deficiency, the payment of \(a\) must also be zero.
    
    \textbf{Case 2.} When \(r(b)=\emptyset\), \(v_a=\eta\), and \(v_b=2\), the item still have to be sold to the agent \(b\). Assume that \(b\) pays \(x_1\leq 2\). By comparing this case with the previous one, we know that the reward that \(a\) received must also be zero, otherwise the agent \(a\) in case 1 will have an incentive to overbid.
    
    \textbf{Case 3.} When \(r(b)=\emptyset\), \(v_a=\eta\), and \(v_b=3\), the item still have to be sold to the agent \(b\). If \(b\)'s price differs between Case 2 and 3, one of them will have the incentive to strategically bid to lower the price, so we are certain that \(b\) still pays \(x_1\), and therefore, has a utility of at least \(3-x_1\geq 3-2=1\).
    
    \textbf{Case 4.} When \(r(b)=\{c\}\), \(v_a=\eta\), \(v_b=3\), and \(v_c=\frac{4}{\eta}\), the item will be allocated to \(c\). The money received by \(b\) must be no lower than \(3-x_1\) as in Case 3, otherwise \(b\) will not diffuse to \(c\).
    
    \textbf{Case 5.} When \(r(b)=\emptyset\), \(v_a=\eta\), and \(v_b=\frac{4}{\eta}\), the item will be allocated to \(b\). By the same reasoning as case 2, \(a\) must have a payment of zero. However, \(a\) can create a Sybil identity \(a_1\) with a bid of \(3\) and insert it between \(a\) and \(b\) (Case 5' in Figure \ref{fig:eff-ratio}), and the situation will look identical to Case 4 for the mechanism. This Sybil attack gives \(a\) an increase in utility of at least 1, contradicting Sybil-proofness.
\end{proof}

For the proof of Theorem~\ref{thm:no-best}, we construct the mechanism \emph{STM with a reserve price \(\kappa\)}, namely \(\mathrm{STM}[\kappa]\), below.

\begin{framed}
 \noindent\textbf{STM with reserve price \(\kappa\)}\\
 \rule{\textwidth}{0.5pt}

\begin{enumerate}
    \item We are given the reported type profile \(\maskprofile =(\theta'_1,\dots,\theta'_m)\) and the neighbor set of the seller \(r(s)\).

    \item We create an auxiliary agent \(p\) with a bid of \(\kappa\) and no outbound edge. Formally, \(\theta'_p=(\kappa, \emptyset)\).
    
    \item\label{step:stmkappa:stm} Perform STM on the type profile \((\theta'_p,\theta'_1,\dots,\theta'_m)\) with the neighbor set of the seller being \(\tilde{r}(s)=\{p\}\cup r(s)\).
    
    \item If STM sells the item to \(p\), \(\mathrm{STM}[\kappa]\) reserves the item
    and sets the payment function of everyone to be zero. Otherwise, \(\mathrm{STM}[\kappa]\) gives the same output as STM in step \ref{step:stmkappa:stm}.
\end{enumerate}
\end{framed}

For every fixed \(\kappa\), \(\mathrm{STM}[\kappa]\) is a non-deficit, Sybil-proof, and IR mechanism. This follows directly from Theorem~\ref{thm:main}.

\begin{lemma}
\label{lem:reserve-price}
For every type profile \(\vek\theta\) and every positive number \(\eps>0\), there exists a non-deficit, SP, IR diffusion auction mechanism \(\mechanism\), such that \(\revenue^{\mechanism}(\vek\theta)+\eps=\welfare^{\mechanism}(\vek\theta)=\welfare^*(\vek\theta)\).
\end{lemma}
\begin{proof}
Let \(\kappa=\welfare^*(\vek\theta)-\eps\), where \(\varepsilon\) is a positive real number small enough that \(\kappa\) is higher than the second-highest bid in \(\vek\theta\). We can see that \(\mathrm{STM}[\kappa]\) satisfies the conditions of the lemma because it is SP and IR, and by its definition, it has to sell the item to the highest bidder with a price of \(\kappa\). Since 
\end{proof}

With this lemma, we can prove the following theorem, combining Theorem~\ref{thm:no-best} and \ref{thm:no-best-rev}.

\begin{theorem}\label{thm:no-best-combined}
   For any non-deficit, SP, IR diffusion auction mechanism \(\mechanism\), for any \(\varepsilon > 0\), there exists another non-deficit, SP, IR diffusion auction mechanism \(\mechanism'\) such that
    \[\inf_{\vek\theta}\frac{\welfare^\mechanism(\vek\theta)}{\welfare^{\mechanism'}(\vek\theta)}<\varepsilon,\text{ and }\inf_{\vek\theta}\frac{\revenue^\mechanism(\vek\theta)}{\revenue^{\mechanism'}(\vek\theta)}<\varepsilon.\]
\end{theorem}
\begin{proof}
By Theorem~\ref{thm:sp-eff-ratio}, \(\mechanism\) has a worst-case efficiency ratio of zero. This means that there exists a type profile \(\vek\theta\) such that \(\welfare^{\mechanism}(\vek\theta)<\frac{\eps}{2} \welfare^*(\vek\theta)\).

By Lemma~\ref{lem:reserve-price}, there exists a  \(\mechanism'=\mathrm{STM}[\kappa]\) such that
\[\revenue^{\mechanism'}(\vek\theta)+\frac{1}{2}\welfare^*(\vek\theta)=\welfare^{\mechanism'}(\vek\theta)=\welfare^*(\vek\theta)\]
Therefore,
\[\revenue^{\mechanism}(\vek\theta)\leq\welfare^{\mechanism}(\vek\theta)<\frac{\varepsilon}{2}\welfare^*(\vek\theta)=\frac{\varepsilon}{2}\welfare^{\mechanism'}(\vek\theta)=\eps\revenue^{\mechanism'}(\vek\theta).\]
Hence,
\[\frac{\welfare^\mechanism(\vek\theta)}{\welfare^{\mechanism'}(\vek\theta)}<\varepsilon,\text{ and }\frac{\revenue^\mechanism(\vek\theta)}{\revenue^{\mechanism'}(\vek\theta)}<\varepsilon,\]
and the infimum will not be higher.
\end{proof}

\end{document}